\documentclass[letterpaper,12pt]{article}

\usepackage{endnotes}
\usepackage[normalem]{ulem}

\usepackage{amssymb}
\usepackage{amsfonts}
\usepackage{amsmath} 
\usepackage{amsthm}
\usepackage{graphicx,color}

\usepackage{comment}
\usepackage{chicago}
\usepackage{array}

\usepackage{url} 

\usepackage{setspace}
\def\R{\mathbb{R}}
\def\endproof{\hfill\diamondsuit}

\def\sA{{\mathcal A}}

\def\sC{{\mathcal C}}

\newcommand\hx{{\hat x}}
\newcommand\tx{{\tilde x}}

\def\tv{\tilde{v}}
\def\ta{\tilde{a}}

\def\E{\mathbb{E}}
\def\V{\mathbb{V}}

\def\N{\mathbb N}

\newcommand{\hp}{\hat{p}}
\newcommand{\hy}{\hat{y}}

\newcommand{\floor}[1]{\left\lfloor #1 \right\rfloor}

\numberwithin{equation}{section}
\theoremstyle{plain}                
\newtheorem{theorem}{Theorem}[section]
\newtheorem{lemma}[theorem]{Lemma}

\theoremstyle{definition}           
\newtheorem{definition}[theorem]{Definition}

\theoremstyle{remark}               

\usepackage[margin=1in]{geometry}

\usepackage{ulem} 

\allowdisplaybreaks

\begin{document}

\pagenumbering{arabic} \pagestyle{plain}

\begin{center}
{\large \bf 
Existence and convergence of discrete-time Kyle models with multiple insiders}

\ \\ \ \\

{\large \bf Jin Hyuk Choi}\\ 
Ulsan National Institute of Science and Technology (UNIST) 

\ \\

{\large \bf Kasper Larsen}\\
Rutgers University

\end{center}
\begin{center}
\ \\

{\normalsize \today }
\end{center}
\vspace{.5cm}

\begin{verse} {\sc Abstract}:  
Foster and Viswanathan (1996) extend the discrete-time setting of Kyle (1985) to multiple informed traders who have partial information about the stock's terminal dividend. We resolve two long-standing open problems in this literature. First, we prove that an  equilibrium exists in the setting of Foster and Viswanathan (1996). Second, as the number of trading times goes to infinity, we prove that the discrete-time equilibrium converges to the continuous-time equilibrium already proven to exist in Back, Cao, and Willard (2000).

\end{verse}

\vspace{0.25cm}









\section{Introduction}  Kyle (1985) is a cornerstone model in market microstructure theory. Holden and Subrahmanyam (1992) extend Kyle (1985) to multiple fully informed traders. Foster and  Viswanathan (1996) further extend Kyle (1985) to multiple  informed traders with less than perfect information. While both Kyle (1985) and Holden and Subrahmanyam (1992) give proofs ensuring that a Nash equilibrium exists, no existence proof is available to guarantee that the model in Foster and  Viswanathan (1996) is supported by a Nash equilibrium. In this paper, we give two mathematical contributions. First, we provide an existence proof for the discrete-time  equilibrium model in Foster and  Viswanathan (1996). Second, as the number of trading times goes to infinity, we prove that the discrete-time equilibrium converges to the continuous-time equilibrium proven to exist in Back, Cao, and Willard (2000).

To reduce the model's complexity, Foster and  Viswanathan (1996) place a symmetry assumption on the traders' information structure.\footnote{Foster and  Viswanathan (1994) give an example of a Kyle model with an asymmetric information structure which leads to a two dimensional Kalman filter. Such extensions are not covered by our analysis.}  The symmetry assumption produces a one dimensional second-order recursion. Because the dimension is one, we can use a shooting technique to prove existence of a solution. In turn, we use this difference solution to construct a discrete-time equilibrium in the setting of Foster and  Viswanathan (1996). 

Similar to Kyle (1985), we subsequently consider convergence from discrete to continuous time.  Back, Cao, and Willard (2000) prove  existence of a certain ODE solution, and show this ODE solution produces a continuous-time equilibrium. We prove that our discrete-time recursion converges pointwise to the ODE solution in Back, Cao, and Willard (2000)  as the number of trading times goes to infinity. Consequently, the discrete-time equilibrium model in Foster and  Viswanathan (1996) converges to the continuous-time model in Back, Cao, and Willard (2000). We note that this type of convergence is not generic in Kyle models. For example, the discrete-time random-horizon model in Caldentey and Stacchetti (2010) does  converge to a well-behaved limit, however, the limit does not correspond to the analogous continuous-time Kyle model.

\section{Foster and Viswanathan (1996)} 
This section outlines the model from Foster and Viswanathan (1996). There are three groups of traders trading the stock at trading times $\{\frac1N,\frac2N,...,1\}$ for some $N\in\N$. A bond with an exogenously specified zero interest rate is also available for trading.

\emph{Noise traders:} At trading time $n\in \{1,...,N\}$, the noise traders'  aggregate stock orders are exogenously given as a Brownian motion increment $\Delta w_n \sim \mathcal{N}(0,\sigma^2)$ where $\sigma^2:=\frac{\sigma_w^2}N $ for a constant $\sigma_w>0$.

\emph{Informed traders:} There are $M\in \N$ partially informed traders. We let $(\ta_1,...,\ta_M)$ be jointly normally distributed random variables with zero means, identical variances $\sigma^2_{\ta} := \E[\ta^2_i]$, and identical covariances $\rho:=\text{cov}(\ta_i,\ta_j)$ for $i,j \in \{1,...,M\}$ with $i\neq j$.  We assume $(\ta_1,...,\ta_M)$ are independent of the noise traders' orders $(\Delta w_1,...,\Delta w_N)$. Informed trader $i \in \{1,...,M\}$ observes $\ta_i$ initially and observes past stock prices $(p_1,...,p_{n-1})$ prior to submitting her own orders $\Delta \theta_{i,n}$ at time $n \in \{1,...,N\}$.\footnote{ As in  Kyle (1985), only $(p_1,...,p_{n-1})$ but not $p_n$ is observed prior to submitting informed orders  $\Delta \theta_{i,n}$ at time $n$.  Rochet and Vila (1994) construct a model where also $p_n$ is observed by the informed trader before she submits orders.} As in Back (1992), there will be a one-to-one correspondence between stock prices $(p_1,...,p_n)$ and aggregate holdings $(y_1,...,y_n)$ defined by the increments
\begin{align}\label{dyn}
\begin{split}
\Delta y_n &:= \sum_ {j=1}^M\Delta \theta_{j,n} + \Delta w_n,\quad  y_0:=0.
\end{split}
\end{align}
Therefore, with no loss of generality, we make the requirement $\Delta \theta_{i,n} \in \sigma(\ta_i,y_1,...,y_{n-1})$. 
\begin{definition}\label{def_admissible} We say $(\Delta \theta_{i,1},...,\Delta \theta_{i,N})\in \sA_{i}$ if and only if $\Delta \theta_{i,1} \in \sigma(\ta_i)$, $\Delta \theta_{i,n} \in \sigma(\ta_i,y_1,...,y_{n-1})$ for $n\in \{2,...,N\}$, and $\E[(\Delta \theta_{i,n})^2]<\infty$ for $n\in \{1,...,N\}$.
$\endproof$
\end{definition}
As in Back, Cao, and Willard (2000), we can with no loss of generality assume   that the stock's terminal value (i.e., the terminal dividends) is exogenously given by 
\begin{align}\label{v}
\tv := \sum_ {i=1}^M\ta_i.
\end{align}
Therefore, informed trader $i$ has partial knowledge of $\tv$ through her initial observation of $\ta_i$. As in Kyle (1985), trader $i$ seeks to solve
\begin{align}\label{obji}
\sup_{\sA_i}\E\Big[\sum_{n=1}^N (\tv - p_n)\Delta \theta_{i,n}\Big| \sigma(\ta_i)\Big].
\end{align}

\emph{Market makers:} At each trading time, the market makers observe aggregate orders and set the stock price as 
\begin{align}\label{me00}
p_n =\E[\tv|\sigma(y_1,...,y_n)],\quad n=1,...,N,\quad p_0:=0.
\end{align}
As in Kyle (1985), the martingale property \eqref{me00} is a zero-expected-profits condition and stems from competition among market makers.  

The following definition is  from Foster and Viswanathan (1996).

\begin{definition}\label{def_Nash} Sequences of non-zero constants $(\lambda_n,\zeta_n,\beta_n)_{n=1}^N$ constitute a discrete-time Nash equilibrium if and only if the following two conditions hold:
\begin{itemize}
\item[(i)] For $n=1,...,N$, the processes 
\begin{align}
\begin{split}\label{hatprocesses}
\Delta \hat{\theta}_{i,n} &:= \beta_n(\ta_i -\hat{t}_{n-1}),\quad \hat{\theta}^i_0:=0, \\
\Delta \hat{y}_n &:= \sum_ {i=1}^M\Delta \hat{\theta}_{i,n} + \Delta w_n,\quad  \hat{y}_0:=0, \\
\Delta \hat{p}_n &:=\lambda_n \Delta \hat{y}_n,\quad \hp_0:=0, \\
\Delta \hat{t}_n &:=\zeta_n\Delta  \hat{y}_n,\quad \hat{t}_0:=0, 
\end{split}
\end{align}
satisfy the market efficiency condition
\begin{align}\label{me}
\hat{p}_n =\E[\tv|\sigma(\hy_1,...,\hy_n)].
\end{align}
\item[(ii)] For $n=1,...,N$ and for any informed trader $i\in \{1,...,M\}$, the market makers' response function $\Delta p_n$ is defined as 
\begin{align}\label{unhatprocesses}
\begin{split}
\Delta p_n &:= \lambda_n \Delta y_n,\quad p_0:=0,\\
\Delta t_n &:= \zeta_n \Delta y_n,\quad t_0:=0,
\end{split}
\end{align}
where $\Delta y_n$ is defined in \eqref{dyn} when informed trader $j$'s response is
\begin{align}\label{responsej}
\Delta \theta_{j,n} = \beta_n (\ta_j - t_{n-1}),\quad j\neq i.
\end{align}

Then, there exist constants $(\alpha_n)_{n=1}^N$ such that trader $i$'s optimal trades for \eqref{obji} are
\begin{align}\label{optimali}
\beta_n (\ta_i - t_{n-1}) + \alpha_n(\hat{t}_{n-1}-t_{n-1}).
\end{align}
$\endproof$
\end{itemize}
\end{definition}

%



Holden and Subrahmanyam (1992) ensure existence of a Nash equilibrium when all informed traders have full information (i.e., when $|\rho| = \sigma_{\ta}^2$). The next result ensures existence of a Nash equilibrium for partially symmetrically informed traders. Its proof is our main contribution and is given in the next section. The following  parameter restriction on the covariance coefficient $\rho$ is equivalent to the variance-covariance matrix being positive definite, hence,  is also in Foster and Viswanathan (1996). We will use the following notation 
\begin{align}
y^{\pm}:= \frac{(M+1)y -1 \pm \sqrt{4y + \big((M+1)y -1\big)^2}}{2}, \quad y>0. \label{y+_def}
\end{align}

\begin{theorem}\label{thm_Main}For $M\ge 2$ partially informed traders, $N\in\N$ trading rounds, assume $-\frac{\sigma_{\ta}^2}{M-1}<\rho < \sigma_{\ta}^2$, and  set $c_0 := (M-1)(\rho- \sigma^2_{\ta})<0$.
\begin{enumerate}

\item There exists a constant $\Sigma_{N} > 0$ such the recursion
\begin{align}\label{recursive}
\begin{split}
&\Sigma_{N-1}: = \Sigma_{N}^+,\\
&(\Sigma_n - \Sigma_{n+1})\Big((1+(M+1)\Sigma_{n-1})\Sigma_n -(1+\Sigma_{n-1})\Sigma_{n-1} \Big)^2\\
&=\Sigma_{n-1}\Sigma_{n+1}(\Sigma_{n-1}-\Sigma_n)\Big(\Sigma_n + (M-1)\Sigma_{n+1}\Big)^2,\quad n=N-1,...,1,
\end{split}
\end{align}
with
\begin{align}\label{Bellman2a}
\begin{split}
\Sigma^+_{n}> \Sigma_{n-1}> \Sigma_{n},\quad n=N-1,...,1,
\end{split}
\end{align}
uniquely produces the initial value $\Sigma_{0} :=\frac{M \sigma^2_{\ta}+c_0}{-c_0}>0$.

\item  The following coefficients produce a  discrete-time  Nash equilibrium
\begin{align}\label{S2intermsofS1}
\begin{split}
\Sigma_{2,n}&:=(M-1) \Sigma_n  \left(\sigma_{\ta}^2-\rho \right),\quad n=0,1,...,N,\\
\beta_n &:=\frac{ \sigma\sqrt{\Sigma_{2,n-1} - \Sigma_{2,n}}}{\sqrt{M}\sqrt{\Sigma_{2,n-1}}\sqrt{\Sigma_{2,n}}},\quad
\zeta_n := \frac{\beta_n\Sigma_{2,n}}{\sigma^2},\quad \lambda_n : = M\zeta_n\quad n=1,...,N.
\end{split}
\end{align}
\end{enumerate}
\end{theorem}
The proof is in Section \ref{sec:Existence}.

\section{Back, Cao, and Willard (2000)} Before stating our convergence result, this section briefly outlines the continuous-time model as it is presented in Back, Cao, and Willard (2000).

\emph{Noise traders:} The noise traders'  aggregate stock holdings are exogenously given as a Brownian motion $w = (w_t)_{t\in[0,1]}$, which we for simplicity of exposition take to be a standard Brownian motion with zero drift and variance $t$. 

\emph{Informed traders:}  The private random variables $(\ta_1,...,\ta_M)$ are as in Foster and Viswanathan (1996) and are assumed to be independent of $w$. Back, Cao, and Willard (2000) restrict insider $i$ to only be allowed to use absolutely continuous holding processes $\theta_i = (\theta_{i,s})_{s\in[0,1]}$. We denote the time derivative of $\theta_{i,s}$ by $\theta'_{i,s}$, which is required to be adapted to $\sigma(\ta_i, (y_u)_{u\in[0,s]})$, where
the aggregate holding process is defined by the dynamics
\begin{align}\label{dyn_BCW1}
\begin{split}
d y_s&:= \sum_ {j=1}^M \theta'_{j,s}ds + d w_s,\quad  y_0:=0.
\end{split}
\end{align}
 
Trader $i$'s optimization problem \eqref{obji} is adjusted to\footnote{To have well-defined optimization problems, additional integrability conditions are imposed in Back, Cao, and Willard (2000) to specify insider $i$'s admissible set $\sA_i$. Because we do not need these conditions in what follows, we omit them.}
\begin{align}\label{obji_BCW1}
\sup_{\theta_i'\in\sA_i}\E\Big[\int_0^1(\tv - p_s) \theta'_{i,s}ds\Big| \sigma(\ta_i)\Big],
\end{align}
where the linear pricing rule in \eqref{unhatprocesses} is adjusted to
\begin{align}
\begin{split}
dp_s &:= \lambda(s)dy_s,\quad p_0:=0,\\
\end{split}
\end{align}
for a deterministic function of time $\lambda = \lambda(s)$.  To solve insider $i$'s optimization problem \eqref{obji_BCW1}, Back, Cao, and Willard (2000) consider $j$'s response to  $i$'s arbitrary order-rate process $\theta'_i$ given by
\begin{align}\label{BCWa1}
\begin{split}
\theta_{j,s}' &= \alpha(s)p_s + \beta(s)\tilde a_j,\\
dp_s &= \lambda(s)\Big(\theta'_{i,s}+  \sum_{j\neq i} \big( \alpha(s)p_s + \beta(s)\tilde a_j\big) + dw_s\Big),\quad p_0:=0.
\end{split}
\end{align}
In \eqref{BCWa1}, the coefficients $\alpha = \alpha(s)$ and $\beta = \beta(s)$ are deterministic functions of time. 

\emph{Market makers:} In equilibrium, the market makers zero-expected profit condition is the property
\begin{align}\label{me0}
\hp_s =\E[\tv|\sigma(\hy_u)_{u\in[0,s]}],\quad s\in[0,1],
\end{align}
where $\hy$ denotes the aggregate-order flow in \eqref{dyn_BCW1} when all insiders use their optimizers, i.e.,
\begin{align}\label{BCWa2}
\begin{split}
d \hy_s&:= \sum_ {j=1}^M \big(\alpha(s)\hp_s + \beta(s)\tilde a_j\big)ds + d w_s,\quad  \hy_0:=0,\\
d\hp_s &:= \lambda(s)\Big(  \sum_{j=1}^M \big( \alpha(s)\hp_s + \beta(s)\tilde a_j\big) + dw_s\Big),\quad \hp_0:=0.
\end{split}
\end{align}

The following result is from Back, Cao, and Willard (2000). 
\begin{theorem}[Back, Cao, and Willard] For $M\ge 2$ partially informed traders and $-\frac{\sigma_{\ta}^2}{M-1}<\rho < \sigma_{\ta}^2$, we have:
\begin{enumerate}
\item The second-order ODE 
\begin{align}\label{BCW1}
\begin{cases}
S''(t) = \frac{4(M - 1)S(t) - 2}{MS(t)^2}S'(t)^2,\quad t\in(0,1), \\
S(0) = \frac{\V[\tv]}{M(M-1)(\sigma^2_{\ta}-\rho)},\quad S(1) = 0,
\end{cases}
\end{align}
has a strictly decreasing solution $S\in \sC([0,1]) \cap \sC^2([0,1))$. 
\item A continuous-time Nash equilibrium is given  by the coefficient functions
\begin{align}\label{BCWS}
 \beta(t) := \sqrt{-\frac{G'(t)}{G(t)^2
}},\quad \lambda(t) :=\beta(t)G(t),\quad \alpha(t) := -\frac1M \beta(t),
\end{align}
where the remaining variance of $\tv$ is $G(t):= M(M-1)( \sigma^2_{\ta}-\rho)S(t)$.
\end{enumerate}
\end{theorem}

\proof The equilibrium existence proof in Back, Cao, and Willard (2000) is based on the ODE in their Eq. (31), which  reads
\begin{align}\label{BCWkeyode}
\frac{d}{dt}\frac1{\lambda(t)}= \beta(t)\Big(2-\frac1M+\frac{\Sigma_2(t)}{\Omega(t)-M\Sigma_2(t)}\Big),
\end{align}
where the coefficients are in \eqref{BCWS} and 
\begin{align}\label{BCW001}
\begin{split}
\Sigma_2:=\frac{G}M,\quad  \Omega :=\Sigma_2\Big( M -\frac{\Sigma_2}{\Sigma_1}\Big), \quad \Sigma_1 :=\frac{\Sigma_2 - (M-1)(\rho -\sigma_{\ta}^2)}M.
\end{split}
\end{align}
Inserting the coefficients \eqref{BCWS} and \eqref{BCW001} into \eqref{BCWkeyode} produces the ODE in \eqref{BCW1}.

$\endproof$

Because all discrete-time equilibrium coefficients  $(\lambda_n,\zeta_n,\beta_n)_{n=1}^N$ in \eqref{S2intermsofS1} are given as continuous functions of the discrete-time recursion \eqref{recursive}, and all continuous-time equilibrium coefficients \eqref{BCWS} and \eqref{BCW001} are given as continuous functions of the ODE solution $S$ of \eqref{BCW1} and its derivative $S'(t)$, all coefficients converge as soon as we show that the recursion converges pointwise to ODE solution. Consequently, the next result shows that the discrete-time equilibrium model Foster and Viswanathan (1996) converges to the continuous-time equilibrium model in Back, Cao, and Willard (2000). 

\begin{theorem}\label{thm_converge} Let $M\ge 2$ and $-\frac{\sigma_{\ta}^2}{M-1}<\rho < \sigma_{\ta}^2$. For $N\in\N$, let $\Sigma_n$, $n \in \{0,1,...,N\}$, be the unique solution of \eqref{recursive}-\eqref{Bellman2a} and let $S(t)$, $t\in[0,1]$, be the unique solution of \eqref{BCW1}. Then, 
for $t\in [0,1)$, we have 
\begin{align}
\lim_{N\to \infty}\Sigma_{\floor{tN}} = S(t) \quad \textrm{and}\quad 
\lim_{N\to \infty}\frac{\Sigma_{\floor{tN}+1}-\Sigma_{\floor{tN}}}{\frac1N} = S'(t). \label{convergence_S_Sprime}
\end{align}
\end{theorem}

The proof is in Section \ref{sec:Convgence}. Theorem \ref{thm_converge}  cannot be applied for $t=1$. The complication is the singularity of the right-hand-side of \eqref{BCW1} at $t=1$ where $S(1) = 0$.

\section{Existence proofs}\label{sec:Existence}
We start by proving that an auxiliary difference equation has a solution, which we then subsequently use to produce our existence proof.

\subsection{An exogenous difference equation} 


\begin{lemma}\label{lem_difference}  In the setting of Theorem \ref{thm_Main}, we let $c_0 := (M-1)(\rho- \sigma^2_{\ta})<0$.
\begin{enumerate}
\item

For a given constant $\Sigma_{N}>0$, there is a unique solution $\Sigma_0, \Sigma_1,...,\Sigma_{N-1} $ of  
\eqref{recursive} satisfying \eqref{Bellman2a}. Furthermore, the solution $\Sigma_0, \Sigma_1,...,\Sigma_{N-1}$ depends continuously on $\Sigma_N$.

\item There exists a constant $\Sigma_{N} > 0$ such that $\Sigma_{0} =\frac{M \sigma^2_{\ta}+c_0}{-c_0}>0$.
\item The following inequality holds for $n=1,2,\dots,N-1$:
\begin{align}
\frac{\Sigma_{n-1}-\Sigma_n}{\Sigma_n-\Sigma_{n+1}} \le \frac{M^4 \Sigma_n^2 \Sigma_{n-1} \Sigma_{n+1}}{(\Sigma_{n}+(M-1)\Sigma_{n+1})^4}. \label{D_ineq}
\end{align}

\item The recursive equation $I_{22,N}=0$ and
$$
I_{22,n-1} =\frac{M \sigma ^2 \big((M-1) \Sigma_{n}+\Sigma_{n-1}\big)^2}{4 c_0 I_{22,n} \Sigma_{n-1} \Sigma_{n} (\Sigma_{n-1}-\Sigma_{n})+4 M^{3/2} \sigma  \Sigma_{n-1}^{3/2} \sqrt{\Sigma_{n}} \sqrt{-c_0 (\Sigma_{n-1}-\Sigma_{n})}}, 
$$
for  $n=1,...,N$, has a unique positive solution that satisfies
\begin{align}\label{SOC1}
I_{22,n} \le \frac{M^{3/2}\sigma\sqrt{\Sigma_{n-1}}}{2\sqrt{-c_0\Sigma_n(\Sigma_{n-1}-\Sigma_n})}, \quad n=1,...,N.
\end{align}

\end{enumerate}
\end{lemma}

\noindent \proof 1: For given  $\Sigma_N>0$, the constant
\begin{align}\label{SigmaM}
\Sigma_{N-1}:=\frac{1}{2} \Big((M+1) \Sigma_N-1+\sqrt{\big((M+1)\Sigma_N-1\big)^2+4\Sigma_N}\Big),
\end{align}
satisfies the first equation in \eqref{recursive} and $\Sigma_{N-1} > \Sigma_N$.

Next, we define functions $h$ and $g$ by
\begin{align}
h(x,y,z)&:=(y-z)g(x,y)^2 - x z (x-y)\Big(y+(M-1)z\Big)^2, \label{h_def}\\
g(x,y)&:=(1+(M+1)x)y-(1+x)x,
\end{align}
for $x,y,z \in \R$. The map $x\mapsto g(x,y)$ is a quadratic function with a negative leading coefficient. 
  For $y>0$, the definition of $y^\pm$ in \eqref{y+_def} gives  the properties $g(y^+,y)=g(y^-,y)=0$, $y^-<0<y<y^+$, and 
\begin{align}
g(y,y) = My^2>0,\quad 
g(x,y)>0 \quad \textrm{for} \quad 0<y\leq x< y^+. \label{g_positive}
\end{align}
For $0<z<y$, the map $x\mapsto h(x,y,z)$ is a fourth degree polynomial with a positive leading coefficient and  satisfies
\begin{align*}
 h(y^-,y,z)&= - y^- z (y^- - y)\big(y+(M-1)z\big)^2<0,\\
h(y,y,z)&=(y-z)M^2y^4 >0,\\
 h(y^+,y,z)&= - y^+ z (y^+ - y)\big(y+(M-1)z\big)^2<0.
\end{align*}  
For $0<z<y$, these inequalities imply that $x\to h(x,y,z)=0$ has exactly one root in each of the four intervals: $(-\infty,y^-)$, $(y^-,y)$, $(y,y^+)$, $(y^+,\infty)$. In particular, the root $\tx = \tx(y,z)\in (y,y^+)$ is simple and ensures $h\big(\tx(y,z),y,z\big)=0$ as well as
\begin{align}
\begin{cases}
h(x,y,z)>0&\textrm{for  } y<x<\tx(y,z),\\
h(x,y,z)<0&\textrm{for  } \tx(y,z)<x<y^+.
\end{cases}
\label{h_negative}
\end{align}
For given $\Sigma_{n}>\Sigma_{n+1}>0$, we define $\Sigma_{n-1}:=\tilde{x}(\Sigma_{n},\Sigma_{n+1})$. Then, the inequalities $\Sigma_n<\Sigma_{n-1}<\Sigma_n^+$ hold.

The point $\Sigma_{N-1}$ defined in \eqref{SigmaM} depends continuously on $\Sigma_N>0$. Theorem 3.5 in Alexanderian (2022) ensures that simple real roots depend continuously on  the polynomial's coefficients. Therefore, $(y,z)\mapsto \tx (y,z)$ is continuous for $0<z<y$. Consequently, 
$\Sigma_{N-2}:= \tx(\Sigma_{N-1},\Sigma_N)$ depends continuously on $\Sigma_N>0$. An induction argument extends this continuity property to all of $\Sigma_0,\Sigma_1,\dots,\Sigma_{N-1}$.\ \\

\noindent 2: As a function of $\Sigma_{N} >0$, the solution $\Sigma_0, \Sigma_1,...,\Sigma_{N-1}$ from the previous step is continuous. To produce $\Sigma_{N}$ such that  $\Sigma_{0}=\frac{M \sigma^2_{\ta}+c_0}{-c_0}$, we consider first 
$\Sigma_{N}:=\frac{M \sigma^2_{\ta}+c_0}{-c_0}$. Because $\Sigma_{n} > \Sigma_{N}$ for all $n<N$, we see that $\Sigma_{0} >\frac{M \sigma^2_{\ta}+c_0}{-c_0}$. 

Because $\Sigma_{N-1}$ is defined as the biggest root of \eqref{recursive}, we have the limit
\begin{align*}
\lim_{\Sigma_{N} \downarrow 0} \Sigma_{N-1} = 0.
\end{align*}
To see 
$$
\lim_{\Sigma_{N} \downarrow 0} \Sigma_{n}=0, \quad n=0,1,...,N-2,
$$ 
we use backward induction. To this end,  we assume 
\begin{align*}
\lim_{\Sigma_{N} \downarrow 0} \Sigma_{N-2} =...=\lim_{\Sigma_{N} \downarrow 0} \Sigma_{n}  =0,
\end{align*}
and we seek to prove $\lim_{\Sigma_{N} \downarrow 0} \Sigma_{n-1}  =0$. The previous step ensured  $\Sigma_{n-1} \in (\Sigma_{n}, \Sigma^+_{n})$. Because $y^+$ defined in \eqref{y+_def} has $\lim_{y\downarrow 0} y^+=0$, we have
$$
\lim_{\Sigma_{N} \downarrow 0} \Sigma_{n-1} =0.
$$

All in all, because $\Sigma_{0}$ depends continuously on $\Sigma_{N}$, the interval $(0, \frac{M \sigma^2_{\ta}+c_0}{-c_0}]$ is contained in the image of the function 
$$
(0, \infty)\ni\Sigma_{N} \to \Sigma_{0},
$$
and the second claim follows.\ \\

\noindent 3: The definition of $y^+$ in \eqref{y+_def} implies the equivalence 
\begin{align}\label{keyineq1}
0<z<y \leq z^+ \quad \iff \quad y>0 \quad \textrm{and} \quad \frac{y(y+1)}{(M+1)y+1}\leq z<y.
\end{align}

Let $n\in \{1,2,\dots,N-1\}$. Using \eqref{recursive}, we can rewrite \eqref{D_ineq} as
\begin{align}\label{hypo} 
\frac{ g(\Sigma_{n-1},\Sigma_{n})^2 }{  \Sigma_{n-1}\Sigma_{n+1}\big(\Sigma_n + (M-1)\Sigma_{n+1}\big)^2   } \le \frac{M^4 \Sigma_n^2 \Sigma_{n-1} \Sigma_{n+1}}{(\Sigma_{n}+(M-1)\Sigma_{n+1})^4}.
\end{align}
By taking square roots and using that $g(\Sigma_{n-1},\Sigma_{n})\ge0$ for  $\Sigma_n<\Sigma_{n-1} \leq  \Sigma_n^+$, the inequality \eqref{hypo} is equivalent to
\begin{align}
\big(\Sigma_n + (M-1)\Sigma_{n+1}\big) g(\Sigma_{n-1},\Sigma_{n}) \leq M^2 \Sigma_n \Sigma_{n-1} \Sigma_{n+1}. \label{hypo1}
\end{align}

To show that \eqref{hypo1} holds, we define the function 
\begin{align}
 f(x,y,z):=\big(y+(M-1)z\big)g(x,y)-M^2 xyz,\quad x\ge y \ge z > 0.
\end{align}
Let $\tx =\tx(y,z)$ be as in the proof of 1. The inequality in \eqref{hypo1} follows as soon as we show
\begin{align}
&f\big(\tx(y,z),y,z\big)\leq 0 \quad \textrm{for} \quad 0<z<y \leq z^+. \label{hypo2}
\end{align}
For $0<z<y$ fixed, the map $\R\ni x\mapsto f(x,y,z)$ is a quadratic function with its legs pointing down and
\begin{align*}
&f(y,y,z)=My^2(y-z)>0,\\
&f(y^+,y,z)=-M^2 y^+ y z<0.
\end{align*}
Therefore, for $0<z<y$ fixed, there exists a unique root $\hx=\hx(y,z)\in (y,y^+)$ satisfying $f(\hat x(y,z),y,z)=0$ and
\begin{align}
f(x,y,z) \leq 0 \quad \textrm{for} \quad \hx(y,z) \leq x<y^+ \textrm{  and  }0<z<y \leq z^+. \label{f_ineq}
\end{align}
The explicit expression of $\hx$ is given by
\begin{align}
\begin{split}
&\hx(y,z):=\Big((M+1)y^2 +(1-M)z-y(1+z) \\
&+ \sqrt{4y(y+(M-1)z)^2 +\big((M+1)y^2 +(1-M)z-y(1+z)\big)^2 }\Big)\big/2\big(y+(M-1)z\big). \label{hx}
\end{split}
\end{align} 
Because of \eqref{f_ineq}, it is enough to prove that
\begin{align}
\hat x(y,z) \leq \tx(y,z) \quad \textrm{for all}\quad 0<z<y \leq z^+ \label{hx<tx}
\end{align}
to justify \eqref{hypo2}.
To prove \eqref{hx<tx},  we use \eqref{h_negative} and $\hat x(y,z)\in (y,y^+)$ to see that it suffices to prove 
\begin{align}
h\big(\hat x(y,z),y,z\big)\geq 0 \quad \textrm{for all}\quad 0<z<y\leq z^+. \label{hypo3}
\end{align}
 Because $f(\hx, y,z)=0$, we have for $0<z<y$
\begin{align}
\begin{split}
h(\hx,y,z)&=(y-z) \left( \frac{M^2 \hx y z}{y+(M-1)z }    \right)^2 - \hx z(\hx -y)\big(y+(M-1)z\big)^2\\
&= \frac{\hx z}{(y+(M-1)z)^2 } \left(  M^4   (y-z) y^2 z \,\hx - (\hx -y)\big(y+(M-1)z\big)^4 \right).
\end{split}
\end{align}
Therefore, to prove \eqref{hypo3}, it suffices to prove
\begin{align}
  M^4   (y-z) y^2 z \,\hx  \geq (\hx -y)\big(y+(M-1)z\big)^4   \quad \textrm{for all}\quad 0<z<y \leq z^+. \label{hypo4}
\end{align}
The explicit expression of $\hx$ in \eqref{hx} gives
\begin{align*}
&\frac{M^4   (y-z) y^2 z }{\big(y+(M-1)z\big)^4}  \frac{\hx}{\hx-y} \\
&=\bigg(M^3 y z \Big( y+(M+1)y^2-z+Mz-yz\\
&\quad + \sqrt{4y(y+(M-1)z)^2 + \left(y(y+My-1)+z-(M+y)z\right)^2} \Big)\bigg)\Big/2\big(y+(M-1)z\big)^4\\
&> \frac{M^3 y z \Big( y+(M+1)y^2-z+Mz-yz+ \left(y(y+My-1)+z-(M+y)z\right)\Big)}{2\big(y+(M-1)z\big)^4}\\
&=\frac{M^3  y^2 z ((M+1)y-z) }{\big(y+(M-1)z\big)^4},\quad 0<z<y \leq z^+,
\end{align*}
where the inequality uses $\sqrt{a+b^2}>b$ for $a>0$ and $b\in \R$. From \eqref{keyineq1}, we see that to prove \eqref{hypo4}, it suffices to prove 
\begin{align}
M^3  y^2 z \big((M+1)y-z\big)-\big(y+(M-1)z\big)^4 \geq 0   \textrm{ for } y>0 \text{ and } \tfrac{y(y+1)}{(M+1)y+1}\leq z<y. \label{hypo5}
\end{align}
The left-hand side of \eqref{hypo5} is a strictly concave function in $z$. Therefore, to prove  \eqref{hypo5}, it suffices to observe
\begin{align*}
&M^3  y^2 z \big((M+1)y-z\big)-\big(y+(M-1)z\big)^4\Big|_{z=\tfrac{y(y+1)}{(M+1)y+1}}\\
& =\frac{M^4 y^4 \Big((1+y)\big(1+(M+1)y\big)^2\big(1+(M+2)y\big)-(1+2y)^4\Big)}{(1+(M+1)y)^4}>0,\\
&M^3  y^2 z ((M+1)y-z)-\big(y+(M-1)z\big)^4\Big|_{z=y} =0.
\end{align*}
Here, the inequality is due to $M\ge2$, which ensures
\begin{align*}
(1+y)(1+(M+1)y)^2(1+(M+2)y)-(1+2y)^4 
&> (1+y)(1+3y)^3-(1+2y)^4\\
&=y\big(2+12y+22y^2+11y^3\big)>0.
\end{align*}

\noindent 4: For $n=N$, the inequality \eqref{SOC1} clearly holds. 
For $n<N$, we use backward induction. Assume that \eqref{SOC1} holds at $n+1$, that is, assume
\begin{align}\label{SOC2}
I_{22,n+1} \le \frac{M^{3/2}\sigma\sqrt{\Sigma_{n}}}{2\sqrt{-c_0\Sigma_{n+1}(\Sigma_{n}-\Sigma_{n+1}})}.
\end{align}
We seek to prove that \eqref{SOC1} holds. Using $I_{22}$'s recursion, the inequality in  \eqref{SOC1} is equivalent to
\begin{align*}
&\frac{M \sigma ^2 \big((M-1) \Sigma_{n+1}+\Sigma_{n}\big)^2}{4 c_0 I_{22,{n+1}} \Sigma_{n} \Sigma_{n+1} (\Sigma_{n}-\Sigma_{n+1})+4 M^{3/2} \sigma  \Sigma_{n}^{3/2} \sqrt{\Sigma_{n+1}} \sqrt{-c_0 (\Sigma_{n}-\Sigma_{n+1})}}\\
&\le\frac{M^{3/2}\sigma\sqrt{\Sigma_{n-1}}}{2\sqrt{-c_0\Sigma_n(\Sigma_{n-1}-\Sigma_n})}.
\end{align*}
Because  $I_{22,{n+1}}$ is assumed to satisfy \eqref{SOC2} and $c_0<0$, it suffices to show the inequality
$$
\frac{(\Sigma_{n}+ (M-1)\Sigma_{n+1})^2}{\sqrt{\Sigma_{n+1}} \sqrt{\Sigma_{n}-\Sigma_{n+1}}}\le \frac{M^2\Sigma_{n} \sqrt{\Sigma_{n-1}}}{\sqrt{\Sigma_{n-1}-\Sigma_{n}}},
$$
which holds by \eqref{D_ineq}.
$\endproof$

\subsection{Kalman filtering}

\begin{lemma}\label{lem_Kalman} In the setting of Theorem \ref{thm_Main}, we  set $c_0 := (M-1)(\rho- \sigma^2_{\ta})<0$, let $(\Sigma_{n})_{n=0}^N$ be as in Lemma \ref{lem_difference}, let $(\Sigma_{2,n})_{n=0}^N$ be defined in \eqref{S2intermsofS1}, and define
\begin{align}\label{S2intermsofS1b}
\Sigma_{1,n} &:=-\frac{c_0(1+\Sigma_{n})}M, \quad \Sigma_{3,n}:= M \Sigma_{1,n} + c_0,\quad n=0,1,...,N.
\end{align}
Then, we have the representations
\begin{align}\label{Sigmas}
\begin{split}
\Sigma_{1,n} &= \E\big[(\ta_i - \hat{t}_n)^2\big],\quad n=0,1,...,N,\\
\Sigma_{2,n} &= \E\Big[(\ta_i - \hat{t}_n)\sum_ {j=1}^M (\ta_j - \hat{t}_n)\Big],\quad n=0,1,...,N,\\
\Sigma_{3,n} &= \E\big[(\tv - \hat{p}_n)(\ta_i - \hat{t}_n)\big],\quad n=0,1,...,N,\\
\hp_n &= \E\big[\tv|\sigma(\hat{y}_1,...,\hy_n)\big] = M\hat{t}_n,\quad n=1,...,N.
\end{split}
\end{align}
\end{lemma}

\proof We start by expressing $\Delta \hat y_n$ in \eqref{hatprocesses} as
\begin{align*}
\Delta \hat{y}_n &= \sum_ {i=1}^M\beta_n(\ta_i -\hat{t}_{n-1})+ \Delta w_n.
\end{align*}

For $n=0$,  the representations in \eqref{Sigmas} hold. We proceed by forward induction and assume \eqref{Sigmas} holds for $n-1$ and prove that \eqref{Sigmas} also holds for $n$.  We have the variance
\begin{align}
 \V[\ta_i - \hat{t}_{n-1} -\Delta \hat{t}_n]&= \V[\ta_i - \hat{t}_{n-1} -\zeta_n \Delta \hat{y}_n] \nonumber\\
 &= \V[\ta_i - \hat{t}_{n-1} -\zeta_n \beta_n \sum_{j=1}^M (\ta_j - \hat{t}_{n-1}) - \zeta_n \Delta w_n] \nonumber\\
&= \Sigma_{1,n-1} + \zeta^2_n\beta_n^2 M \Sigma_{2,n-1} -2\zeta_n\beta_n \Sigma_{2,n-1} + \zeta_n^2\sigma^2 \nonumber\\
& = \Sigma_{1,n}, \label{Sigma1}
\end{align}
where the last equality follows from the expressions of $(\Sigma_{2,n-1}, \beta_n,\zeta_n)$ in \eqref{S2intermsofS1}. 
Similarly, we have the covariance 
\begin{align}
& \E\Big[\big(\ta_i - \hat{t}_{n-1} -\Delta \hat{t}_n\big)\times \sum_{k=1}^M\big(\ta_k - \hat{t}_{n-1} - \Delta \hat t_n\big)\Big]  \nonumber  \\
&= \E\Big[\Big(\ta_i - \hat{t}_{n-1} -\zeta_n \beta_n \sum_{j=1}^M (\ta_j - \hat{t}_{n-1}) - \zeta_n \Delta w_n\Big)  \nonumber  \\
&\quad \quad \times \sum_{k=1}^M\Big(\ta_k - \hat{t}_{n-1} -\zeta_n \beta_n \sum_{l=1}^M (\ta_l - \hat{t}_{n-1}) - \zeta_n \Delta w_n\Big)\Big]  \nonumber \\
&= \E\Big[\Big(\ta_i - \hat{t}_{n-1} -\zeta_n \beta_n \sum_{j=1}^M (\ta_j - \hat{t}_{n-1}) - \zeta_n \Delta w_n\Big)  \nonumber \\
&\quad \quad \times\Big( \sum_{k=1}^M(\ta_k - \hat{t}_{n-1}) -\zeta_n M \beta_n \sum_{l=1}^M (\ta_l - \hat{t}_{n-1}) - \zeta_n M \Delta w_n\Big)\Big]  \nonumber \\
&= \Sigma_{2,n-1} -\zeta_n M \beta_n \Sigma_{2,n-1} -\zeta_n M \beta_n \Sigma_{2,n-1} +\zeta^2_n M^2 \beta^2_n \Sigma_{2,n-1} + M\zeta^2_n\sigma^2 \nonumber\\
&=\Sigma_{2,n}, \label{Sigma2}
\end{align}
where the last equality follows from the expressions of $(\Sigma_{2,n-1}, \beta_n,\zeta_n)$ in \eqref{S2intermsofS1}. 
The representation of $\Sigma_{3,n}$ follows from \eqref{v}. 

To compute the pricing coefficients,  we start with the variance of $\Delta \hat y_n$. We have
\begin{align*}
\V[ \Delta \hy_n ] &= \beta_n^2\V\Big[\sum_{j=1}^M (\ta_j - \hat{t}_{n-1})\Big] + \sigma^2\\
&= \beta_n^2 M\Sigma_{2,n-1}+ \sigma^2.
\end{align*}
The dynamics of the price process are
\begin{align}
\Delta \hp_n &= \E[\tv |\sigma(\hy_1,...,\hy_{n})]- \hat{p}_{n-1} \nonumber \\
&= \E[\tv - \hat{p}_{n-1} |\sigma(\Delta\hy_1,...,\Delta\hy_{n})] \nonumber\\
&= \E[\tv - \hat{p}_{n-1} |\sigma(\Delta\hy_{n})] \nonumber\\
&= \frac{\E[(\tv - \hat{p}_{n-1})\Delta \hy_n]}{\V[\Delta \hy_n]}\Delta\hy_n \nonumber\\
&= \frac{\beta_n\E[(\tv - \hat{p}_{n-1})\sum_{j=1}^M (\ta_j - \hat{t}_{n-1})]}{\beta_n^2 M\Sigma_{2,n-1}+ \sigma^2} \Delta\hy_n \nonumber\\
&= \frac{\beta_n M\Sigma_{3,n-1}}{\beta_n^2 M\Sigma_{2,n-1}+ \sigma^2}\Delta\hy_n. \label{dhpn}
\end{align}
In \eqref{dhpn}, the third equality follows from the mutual independence of $(\Delta\hy_1,...,\Delta\hy_{n-1})$ and $(\Delta \hy_n,\tv - \hat{p}_{n-1})$. The fourth equality in \eqref{dhpn} follows from the projection theorem for Gaussian random variables. Inserting $\Sigma_3 = \Sigma_2$ and $\beta_n$ from \eqref{S2intermsofS1} into the last line in \eqref{dhpn} produces 
$$
 \frac{\beta_n M\Sigma_{3,n-1}}{\beta_n^2 M\Sigma_{2,n-1}+ \sigma^2} = M\frac{\beta_n\Sigma_{2,n}}{\sigma^2} = \lambda_n.
$$
In a similar manner, we find the representation 
\begin{align}\label{FVzeta}
\Delta \hat{t}_n = \zeta_n \Delta \hy_n,\quad \zeta_n = \frac{\beta_n \Sigma_{2,n-1}}{\beta_n^2 M\Sigma_{2,n-1}+ \sigma^2}.
\end{align}
$\endproof$

Informed trader $i$'s innovation process $w_{i,n}$ is defined by $w_{i,0}:=0$ and the below Gaussian increments \eqref{win0}. 
For $n=1,\dots,N$ and $i=1,\dots,M$, we define $\mathcal{F}_n^i$ and $\hat{\mathcal{F}}_n^i$ as
\begin{align}
\mathcal{F}_n^i := \sigma(\ta_i,y_1,...,y_{n-1}), \quad \hat{\mathcal{F}}_n^i :=\sigma(\ta_i,\hat{y}_1,...,\hat{y}_{n-1}).   \label{FF_def}
\end{align}

\begin{lemma} Assume the setting of Lemma \ref{lem_Kalman}. For $i,j =1,...,M$ with $i\neq j$, let $j$'s response to $i$ be as in \eqref{responsej}. The Gaussian random variables
\begin{align}\label{win0}
\begin{split}
\Delta w_{i,n} &:=\Delta \hat{y}_n - \beta_n\frac{\Sigma_{2,n-1}}{\Sigma_{1,n-1}}(\ta_i - \hat{t}_{n-1}),\quad n=1,...,N,
\end{split}
\end{align}
are mutually independent and independent of $\hat{\mathcal{F}}_{n}^i$. Furthermore, for $n=1,...,N$, we have
\begin{align}\label{indpendent1}
\begin{split}
\mathcal{F}_{n}^i=\hat{\mathcal{F}}_{n}^i&=\sigma(\ta_i-\hat{t}_{n-1},\hat{t}_1,...,\hat{t}_{n-1}),\\
\ta_j - \hat{t}_{n-1} \perp \hat{t}_k& \quad \textrm{and}\quad \ta_i-\hat{t}_{n-1} \perp \hat{t}_k,\quad \textrm{for}\quad k\le n-1.
\end{split}
\end{align}
Finally, for arbitrary $\Delta \theta_{i,n} \in \mathcal{F}_{n}^i$, the state processes
\begin{align}\label{Y1Y2}
\begin{split}
Y_{1,n} &:= \tilde{a}_i - \hat{t}_n,\quad Y_{2,n} := \hat{t}_n-t_n,\quad n=0,...,N, 
\end{split}
\end{align}
have the following Markovian dynamics 
\begin{align}\label{markov_dyn}
\begin{split}
\Delta Y_{1,n} 
&= - \zeta_n\Big(\Delta w_{i,n} +\beta_n\frac{\Sigma_{2,n-1}}{\Sigma_{1,n-1}}Y_{1,n-1}\Big),\\
\Delta Y_{2,n} 
&=\zeta_n\Big(\beta_nY_{1,n-1} - \beta_n (M-1) Y_{2,n-1} -\Delta \theta_{i,n}\Big).
\end{split}
\end{align}

\end{lemma}

\proof To see the second $\sigma$-algebra equality in \eqref{indpendent1}, we use $\zeta_k \neq 0$ to see that
$$
\sigma(\ta_i,\Delta \hat{y}_1,...,\Delta \hat{y}_{n-1})=\sigma(\ta_i,\Delta \hat{t}_1,...,\Delta \hat{t}_{n-1}),
$$
which implies the claim. The independence properties in \eqref{indpendent1} follow from iterated expectations 
\begin{align*}
\E[(\ta_j - \hat{t}_{n-1})\hat{t}_{k}] &= \E\Big[\E\big[(\ta_j - \hat{t}_{n-1})\hat{t}_{k}|\sigma(\hy_1,...,\hy_{n-1})\big]\Big] \\
&= \E\Big[\hat{t}_{k}\E\big[\ta_j - \hat{t}_{n-1} |\sigma(\hy_1,...,\hy_{n-1})\big]\Big]\\
&=0,
\end{align*}
where the last equality uses that Lemma \ref{lem_Kalman} gives $\hat t_n = \E[\ta_j|\sigma(\hy_1,...,\hy_n)]$. For $j\neq i$, these properties and the projection theorem for Gaussian random variables give
\begin{align}
\sum_{j=1}^M\E[\ta_j - \hat{t}_{n-1}|\hat{\mathcal{F}}_{n}^i]&=\sum_{j=1}^M\E[\ta_j - \hat{t}_{n-1}|\sigma(\ta_i-\hat{t}_{n-1},\hat{t}_1,...,\hat{t}_{n-1})] \nonumber\\
&=\sum_{j=1}^M\E[\ta_j - \hat{t}_{n-1}|\sigma(\ta_i-\hat{t}_{n-1})]  \nonumber\\
&=\frac{\Sigma_{2,n-1}}{\Sigma_{1,n-1}}(\ta_i - \hat{t}_{n-1}).  \label{projection_Gaus1}
\end{align}

Informed trader $i$'s innovation process  \eqref{win0} has the representation
\begin{align}
\beta_n\sum_{j=1}^M\Big( \ta_j - \hat{t}_{n-1} - \E\big[\ta_j - \hat{t}_{n-1}|\hat{\mathcal{F}}_{n}^i\big]\Big) +\Delta w_n &= \Delta \hat{y}_n - \beta_n\sum_{j=1}^M \E\big[\ta_j - \hat{t}_{n-1}|\hat{\mathcal{F}}_{n}^i\big] \nonumber\\
&= \Delta \hat{y}_n - \beta_n\frac{\Sigma_{2,n-1}}{\Sigma_{1,n-1}}(\ta_i - \hat{t}_{n-1}) \nonumber \\
&=\Delta w_{i,n}. \label{win}
\end{align}
Therefore, iterated expectations give for $0< m <n \le N$
\begin{align}
\E[\Delta w_{i,n}\Delta w_{i,m}] &= \E\big[\E[\Delta w_{i,n}\Delta w_{i,m}|\hat{\mathcal{F}}_{m+1}^i]\big] \nonumber\\
&= \E\big[\Delta w_{i,m}\E[\Delta w_{i,n}|\hat{\mathcal{F}}_{m+1}^i]\big] \nonumber\\
&=0.  \label{win2}
\end{align}
Furthermore, joint normality and \eqref{win2} give the independence $\Delta w_{i,n} \perp ( \Delta w_{i,1} ,...,\Delta w_{i,n-1} )$. For $A \in \{\ta^i, \hy_1,...,\hy_{n-1}\}$, iterated expectations give the zero-covariance
\begin{align*}
\E[\Delta w_{i,n} A] &= \E\Big[\E[\Delta w_{i,n} A|\hat{\mathcal{F}}_{n}^i]\Big] \\
&= \E\Big[A \, \E[\Delta w_{i,n} |\hat{\mathcal{F}}_{n}^i]\Big]\\
&=0.
\end{align*}
Then, the independence $\Delta w_{i,n} \perp \hat{\mathcal{F}}_{n}^i$ follows from joint normality. 

We show the first $\sigma$-algebra equality in \eqref{indpendent1} using induction. The basis step holds because
\begin{align*}
\sigma(\ta_i,y_1) &= \sigma(\ta_i, \beta_1 \sum_{j=1,j\neq i}^M  \ta_j +\Delta \theta_{i,1}+\Delta w_1)\\
&= \sigma(\ta_i, \beta_1 \sum_{j=1,j\neq i}^M  \ta_j +\Delta w_1)\\
&=\sigma(\ta_i,\hy_1).
\end{align*}
The next step follows similarly: 
\begin{align*}
\sigma(\ta_i,\hy_1,\hy_2) & =\sigma(\ta_i,\Delta\hat{y}_1,\Delta \hy_2) \\
& =\sigma(\ta_i,\beta_1\sum_{j=1}^M \tilde{a}_j+\Delta w_1,\beta_2\sum_{j=1}^M ( \tilde{a}_j-\hat{t}_1)+\Delta w_2) \\
& =\sigma(\ta_i,\beta_1\sum_{j=1,j\neq i}^M \tilde{a}_j+\Delta w_1,\beta_2\sum_{j=1,j\neq i}^M \tilde{a}_j+\Delta w_2).
\end{align*}
On the other hand, we have
\begin{align*}
\sigma(\ta_i,y_1,y_2) & =\sigma(\ta_i,\Delta y_1,\Delta y_2) \\
& =\sigma(\ta_i,\beta_1\sum_{j=1,j\neq i}^M \tilde{a}_j+\Delta \theta_{i,1}+\Delta w_1,\beta_2\sum_{j=1,j\neq i}^M ( \tilde{a}_j-t_1)+\Delta \theta_{i,2}+\Delta w_2) \\
& =\sigma(\ta_i,\beta_1\sum_{j=1,j\neq i}^M \tilde{a}_j+\Delta w_1,\beta_2\sum_{j=1,j\neq i}^M \tilde{a}_j+\Delta w_2).
\end{align*}
The general case is similar.

Finally, the Markovian dynamics \eqref{markov_dyn} follow from trader $j$'s response \eqref{responsej}. 
\begin{align*}
\Delta Y_{1,n} &= - \Delta \hat{t}_n= - \zeta_n\Delta \hat{y}_n\\
&= - \zeta_n\Big(\Delta w_{i,n} +\beta_n\frac{\Sigma_{2,n-1}}{\Sigma_{1,n-1}}Y_{1,n-1}\Big),\\
\Delta Y_{2,n} &= \Delta \hat{t}_n- \Delta t_n= \zeta_n (\Delta \hat{y}_n- \Delta y_n)\\
&=\zeta_n\Big(\beta_n(\ta^i - \hat{t}_{n-1}) + \beta_n (M-1) (t_{n-1} - \hat{t}_{n-1}) -\Delta \theta^i_n\Big)\\
&=\zeta_n\Big(\beta_nY_{1,n-1} - \beta_n (M-1) Y_{2,n-1} -\Delta \theta^i_n\Big).
\end{align*}

$\endproof$

\subsection{Remaining proof of Theorem \ref{thm_Main}}

\emph{Proof of Theorem \ref{thm_Main}:} Let $I_{22,n}$ be as in Lemma \ref{lem_difference}, let $(\Sigma_{n}, \Sigma_{1,n},\Sigma_{2,n},\Sigma_{3,n}, \beta_n,\zeta_n,\lambda_n)$ be as in Lemma \ref{lem_Kalman}, and define the state processes $Y_1$ and $Y_2$ as in \eqref{markov_dyn}.  The inequality in \eqref{SOC1} ensures\footnote{Because we divide by 2 in  \eqref{SOC1}, we have strict inequality in \eqref{realSOC}.}
\begin{align}\label{realSOC}
 I_{22,n} \zeta_n<M.
\end{align}
Therefore, the following recursions are well-defined
\begin{align*}
I_{0,n-1} &:=  I_{0,n}+ I_{11,n} \V[ \Delta w_{i,n} ] \zeta_n^2,\\
I_{11,n-1} &:= \Big(4 \zeta_n I_{11,n} (M-\zeta_n I_{22,n}) (\Sigma_{1,n-1}-\beta_n \zeta_n \Sigma_{2,n-1})^2-2 \zeta_n \Sigma_{1,n-1} \Sigma_{2,n-1}\\
&\quad \times \left(\beta_n \zeta_n^2 I_{12,n}^2+\beta_n \zeta_n I_{12,n} M (\beta_n \zeta_n-1)+I_{12,n}+\beta_n (M-2 \zeta_n I_{22,n}) (\beta_n \zeta_n M-1)\right)\\
&\quad+\zeta_n^2 \Sigma_{1,n-1}^2 (I_{12,n}+\beta_n M)^2+(\beta_n \zeta_n \Sigma_{2,n-1} (\zeta_n I_{12,n}-M)+\Sigma_{2,n-1})^2\Big)\\
&\quad \Big/4 \zeta_n \Sigma_{1,n-1}^2 (M-\zeta_n I_{22,n}),\\
I_{12,n-1} &:= \bigg(\big(1-\beta_n \zeta_n (M-1)\big) \Big(\Sigma_{2,n-1} \big(-\beta_n \zeta_n M \big(\zeta_n (I_{12,n}-2 I_{22,n})+M\big)\\
&\quad -2 \zeta_n I_{22,n}+M\big)+\zeta_n M \Sigma_{1,n-1} (I_{12,n}+\beta_n M)\Big)\bigg)\Big/{2 \zeta_n \Sigma_{1,n-1} (M-\zeta_n I_{22,n})},
\end{align*}
with terminal values 
\begin{align}\label{terminalconds}
I_{0,N}:=I_{11,N}:=I_{12,N}:=0.
\end{align}
The variance of \eqref{win0} appearing of the right-hand side of $I_{0,n-1}$ above is
\begin{align*}
\V[ \Delta w_{i,n} ] &= \V\Big[ \Delta \hy_n -\beta_n\frac{\Sigma_{2,n-1}}{\Sigma_{1,n-1}}(\ta_i -\hat{t}_{n-1}) \Big]\\
&=\sigma^2 +\beta_n^2\V\Big[ \sum_{j=1}^M (\ta_j - \hat{t}_{n-1})-\frac{\Sigma_{2,n-1}}{\Sigma_{1,n-1}}(\ta_i -\hat{t}_{n-1}) \Big]\\
&=\sigma^2 +\beta_n^2\bigg(M\Sigma_{2,n-1}-\frac{\Sigma_{2,n-1}^2}{\Sigma_{1,n-1}}\bigg).
\end{align*}

Next, we prove that the value function corresponding to \eqref{obji} has the quadratic form
\begin{align}\label{FVval}
\begin{split}
\sup_{\substack{
  \Delta \theta_{i,k} \in \mathcal{F}_{k}^i \\
  \E[(\Delta \theta_{i,k})^2] < \infty
}}&\sum_{k=n+1}^N \E[ (\tv -p_k)\Delta \theta_{i,k}|\mathcal{F}_{n}^i] = I_{0,n}+I_{11,n} Y^2_{1,n}+I_{12,n} Y_{1,n}Y_{2,n}+I_{22,n} Y^2_{2,n}.
\end{split}
\end{align}
To see this, we let $\Delta \theta_{i,n} $ be arbitrary and compute
\begin{align}
\E[\Delta \hp_n - \Delta p_n | \mathcal{F}_{n}^i]&=\lambda_n\E\Big[\sum_{j=1}^M \Delta \hat\theta_{j,n} - \sum_{j=1,j\neq i}^M \Delta \theta_{j,n} - \Delta \theta_{i,n} \big|  \mathcal{F}_{n}^i  \Big] \nonumber \\
&=\lambda_n\Big(\E\Big[\sum_{j=1}^M \beta_n(\ta_j - \hat{t}_{n-1}) - \sum_{j=1,j\neq i}^M \beta_n(\ta_j - t_{n-1})  \big|  \mathcal{F}_{n}^i\Big]-\Delta \theta_{i,n}\Big) \nonumber \\
&=\lambda_n\Big(\beta_n(\ta_i - \hat{t}_{n-1}) + \beta_n (M-1) (t_{n-1} - \hat{t}_{n-1}) -\Delta \theta_{i,n}\Big).  \label{p1}
\end{align}
The definition of trader $i$'s innovation process $w_{i,n}$ in \eqref{win0} produces the representation
\begin{align}\label{p2}
\begin{split}
\Delta \hp_n &= \lambda_n\Delta \hy_n\\
&= \lambda_n\Big(\Delta w_{i,n} +\beta_n\frac{\Sigma_{2,n-1}}{\Sigma_{1,n-1}}(\ta_i -\hat{t}_{n-1})\Big).
\end{split}
\end{align}
We use \eqref{p2} to calculate
\begin{align}\label{dpp1a}
\begin{split}
&\E[(\tv-p_n) \Delta \theta_{i,n} | \mathcal{F}_{n}^i]=\Delta \theta_{i,n}\E[(\tv-\hp_{n-1}+\hp_{n-1}-p_{n-1}-\Delta \hp_n+\Delta \hp_n-\Delta p_n)  |  \mathcal{F}_{n}^i],
\end{split}
\end{align}
as follows:
\begin{align*}
\E[\tv-\hp_{n-1}|   \mathcal{F}_{n}^i ] & =\frac{\Sigma_{3,n-1}}{\Sigma_{1,n-1}}(\ta_i - \hat{t}_{n-1}),\\
\E[\hp_{n-1}-p_{n-1}|   \mathcal{F}_{n}^i ] & =M(\hat{t}_{n-1}-t_{n-1}),\\
\E[\Delta \hp_{n}|   \mathcal{F}_{n}^i ] & = \lambda_n\beta_n\frac{\Sigma_{2,n-1}}{\Sigma_{1,n-1}}(\ta_i -\hat{t}_{n-1}).
\end{align*}
By combining these expressions with  \eqref{p1} and using $\Sigma_2=\Sigma_3$, we see that \eqref{dpp1a} becomes
\begin{align}
&\E[(\tv-p_n) \Delta \theta_{i,n} | \mathcal{F}_{n}^i ] \nonumber\\
&=\Delta \theta_{i,n}\bigg(\frac{\Sigma_{2,n-1}}{\Sigma_{1,n-1}}(\ta_i - \hat{t}_{n-1})+M(\hat{t}_{n-1}-t_{n-1})-\lambda_n\beta_n\frac{\Sigma_{2,n-1}}{\Sigma_{1,n-1}}(\ta_i -\hat{t}_{n-1}) \nonumber\\
&\quad +\lambda_n\Big(\beta_n(\ta_i - \hat{t}_{n-1}) + \beta_n (M-1) (t_{n-1} - \hat{t}_{n-1}) -\Delta \theta_{i,n}\Big)\bigg) \nonumber\\
&= \Delta \theta_{i,n}\Big(\frac{\Sigma_{2,n-1}}{\Sigma_{1,n-1}}(1-\lambda_n\beta_n)+\lambda_n\beta_n\Big)Y_{1,n-1} \nonumber\\
&\quad+\Delta \theta_{i,n}\big(M-\lambda_n\beta_n(M-1)\big)Y_{2,n-1}-\lambda_n (\Delta \theta_{i,n})^2. \label{dpp1b}
\end{align}
Based on \eqref{dpp1b}, we have
\begin{align}
&\E[ (\tv-p_n) \Delta \theta_{i,n} +  I_{0,n}+I_{11,n} Y^2_{1,n}+I_{12,n} Y_{1,n}Y_{2,n}+I_{22,n} Y^2_{2,n}| \mathcal{F}_{n}^i] \nonumber\\
&= \Delta \theta_{i,n}\Big(\frac{\Sigma_{2,n-1}}{\Sigma_{1,n-1}}(1-\lambda_n\beta_n)+\lambda_n\beta_n\Big)Y_{1,n-1} \nonumber\\
&\quad +\Delta \theta_{i,n}\big(M-\lambda_n\beta_n(M-1)\big)Y_{2,n-1}-\lambda_n (\Delta \theta_{i,n})^2 \nonumber \\
&\quad + I_{0,n}+I_{11,n}  \zeta^2_n\V[\Delta w_{i,n}]+ I_{11,n} \Big(Y_{1,n-1}  - \zeta_n\beta_n\frac{\Sigma_{2,n-1}}{\Sigma_{1,n-1}}Y_{1,n-1}\Big)^2 \nonumber\\
&\quad +I_{12,n}  \Big(Y_{1,n-1}  - \zeta_n\beta_n\frac{\Sigma_{2,n-1}}{\Sigma_{1,n-1}}Y_{1,n-1}\Big)\Big(Y_{2,n-1}+\zeta_n\Big(\beta_nY_{1,n-1} - \beta_n (M-1) Y_{2,n-1} -\Delta \theta_{i,n}\Big)\Big) \nonumber\\
&\quad +I_{22,n} \Big(Y_{2,n-1}+\zeta_n\Big(\beta_nY_{1,n-1} - \beta_n (M-1) Y_{2,n-1} -\Delta \theta_{i,n}\Big)\Big)^2. \label{dpp10}
\end{align}
The second-order condition for \eqref{dpp10} is given by \eqref{realSOC}. Consequently, the optimizer for \eqref{dpp10} is given by
\begin{align}\label{optimize1}
\begin{split}
&\frac{\zeta_n \Sigma_{1,n-1} (I_{12,n}+2 \beta_n \zeta_n I_{22,n}-\beta_n M)+\Sigma_{2,n-1} (\beta_n \zeta_n (M-\zeta_n I_{12,n})-1)}{2 \zeta_n \Sigma_{1,n-1} (\zeta_n I_{22,n}-M)}Y_{1,n-1}\\
&+\frac{(M-2 \zeta_n I_{22,n}) (\beta_n \zeta_n (M-1)-1)}{2 \zeta_n (\zeta_n I_{22,n}-M)}Y_{2,n-1}.
\end{split}
\end{align}

Finally, we turn to establish the equilibrium properties. To get equations for $\beta_n$ and $\alpha_n$, we re-express \eqref{optimali} as
\begin{align}\label{optimaliAA}
\begin{split}
\beta_n (\ta_i - t_{n-1}) + \alpha_n(\hat{t}_{n-1}-t_{n-1}) &= \beta_n Y_{1,n-1} +(\beta_n+\alpha_n) Y_{2,n-1}.
\end{split}
\end{align}
Matching coefficients in \eqref{optimize1} and \eqref{optimaliAA} produces the requirements
\begin{align}\label{optimizeAA1}
\begin{split}
\beta_n&=\frac{\zeta_n \Sigma_{1,n-1} (I_{12,n}+2 \beta_n \zeta_n I_{22,n}-\beta_n M)+\Sigma_{2,n-1} \big(\beta_n \zeta_n (M-\zeta_n I_{12,n})-1\big)}{2 \zeta_n \Sigma_{1,n-1} (\zeta_n I_{22,n}-M)},\\
\beta_n+\alpha_n&=\frac{(M-2 \zeta_n I_{22,n}) \big(\beta_n \zeta_n (M-1)-1\big)}{2 \zeta_n (\zeta_n I_{22,n}-M)}.
\end{split}
\end{align}
To see that the first equation in \eqref{optimizeAA1} holds, we insert the expressions in \eqref{S2intermsofS1} and solve for $I_{12,n}$ to get
\begin{align}\label{I12n0}
\begin{split}
I_{12,n} &= \frac{\sqrt{M} \sigma  \sqrt{\Sigma_{n-1}} \big(\Sigma_{n-1} (M \Sigma_{n}-\Sigma_{n-1}+\Sigma_{n}-1)+\Sigma_{n}\big)}{\Sigma_{n-1} (\Sigma_{n}+1) \sqrt{\Sigma_{n}} \sqrt{-c_0(\Sigma_{n-1}-\Sigma_{n})}}.
\end{split}
\end{align}
The terminal condition $I_{12,N}=0$ in \eqref{terminalconds} holds because of the first property in  \eqref{recursive}. Inserting $I_{12,n-1}$ and $I_{12,n}$ from \eqref{I12n0} into the above recursion for $I_{12,n}$ produces the requirement
\begin{align}\label{I12nnn}
\begin{split}
\frac{\sqrt{\Sigma_{n-2}} \big(\Sigma_{n-2} (M \Sigma_{n-1}+\Sigma_{n-1}-\Sigma_{n-2}-1)+\Sigma_{n-1}\big)}{\Sigma_{n-2}  \sqrt{(\Sigma_{n-2}-\Sigma_{n-1})}}=\frac{\sqrt{\Sigma_{n}} \big((M-1) \Sigma_{n}+\Sigma_{n-1}\big)}{\sqrt{(\Sigma_{n-1}-\Sigma_{n})}}.
\end{split}
\end{align}
By squaring both sides, we see that \eqref{I12nnn} holds because of the recursion in \eqref{recursive}. 

Finally, we use the second equation in \eqref{optimizeAA1} to define $\alpha_n$ as 
\begin{align}
\alpha_n&:=\frac{(M-2 \zeta_n I_{22,n}) \big(\beta_n \zeta_n (M-1)-1\big)}{2 \zeta_n (\zeta_n I_{22,n}-M)} - \beta_n.
\end{align}

$\endproof$

\section{Convergence proofs}\label{sec:Convgence}

\subsection{Recursion estimates}

Throughout this subsection, we set $\Sigma_{N+1}:=0$, which is consistent with the recursion in \eqref{recursive}.

The next result proves that the solution to the recursion in \eqref{recursive} is concave for $n$ sufficiently big and provides a uniform bound on the discrete derivative at $n=0$ in the sense that $\frac{\Sigma_1 - \Sigma_0}{1/N}=-N(\Sigma_0 - \Sigma_1)$ is uniformly bounded with respect to $N$. We are not claiming that the bound $\Sigma_n < \frac{1}{2M^2}$ for concavity is necessary. Indeed, from the ODE \eqref{BCW1}, we see that $S(t) < \frac{1}{2(M-1)}$ is both necessary and sufficient for the continuous-time limiting model.

\begin{lemma}\label{bN_bound} Assume the setting of Theorem \ref{thm_Main} and let $\Sigma_0,...,\Sigma_N$ be as in Theorem \ref{thm_Main}.1. Then, we have
\begin{enumerate}
\item  For $\Sigma_n< \frac{1}{2M^2}$, we have concavity in the sense that
\begin{align}
\Sigma_{n-1} - \Sigma_n < \Sigma_n - \Sigma_{n+1}. \label{difference_concavity}
\end{align}
\item For $N >  2M^2 \Sigma_0-1$, we have the uniform bound
\begin{align}
0<N(\Sigma_0 - \Sigma_1) \leq  \max\left\{ \Sigma_0, (4M^2)^{4(M-1)/M} \Sigma_0^{(5M-4)/M} \right\}. \label{derivative_bound}
\end{align}
\end{enumerate}
\end{lemma}

\proof 1. We split into two cases. First, we consider $n=N$. The inequality $\Sigma_{N-1}-\Sigma_N<\Sigma_N - \Sigma_{N+1}$ reduces to $\Sigma_N^+ < 2\Sigma_N$ because $\Sigma_{N-1}=\Sigma_N^+$. This is equivalent to $0<\Sigma_N<\frac{1}{2(M-1)}$. Since $M\geq2$, $\Sigma_N< \frac{1}{2M^2}$ implies the inequality. 

Second, we consider $n\in \{1, \dots, N-1\}$ and suppose that $\Sigma_n< \frac{1}{2M^2}$. We start with the auxiliary inequality
\begin{align}\label{aux_ineq1}
y - \frac{y(1+y)}{(M+1)y+1}<y^+ - y,\quad y>0.
\end{align}
Inserting $y^+$ in \eqref{y+_def} gives us the equivalent inequality
$$
y\big(y+2-(M-2) M y\big)+1< (M y+y+1) \sqrt{(M y+y-1)^2+4 y},\quad y>0.
$$
This inequality clearly holds if the left-hand side is negative. Alternatively, if the left-hand side is positive, we square both sides to obtain the equivalent inequality $M^2 y (2 M y+y+2)>0$, which trivially holds. 

Based on \eqref{aux_ineq1} and the equivalence of $0<\Sigma_n \leq \Sigma_{n+1}^+$ and $\Sigma_{n+1} \geq  \frac{\Sigma_{n} (1+\Sigma_{n} )}{(M+1)\Sigma_{n} +1}$, we have
\begin{align}
\Sigma_{n} - \Sigma_{n+1}\le  \Sigma_{n}  -  \frac{\Sigma_{n} (1+\Sigma_{n} )}{(M+1)\Sigma_{n} +1} < \Sigma_n^+ - \Sigma_{n} .  \label{ys}
\end{align}

For $y>0$ and $a\in (0,y]$, we have
\begin{align}
h(y+a,y,y-a)&=a^2 \tilde h(y,a), \label{th}
\end{align}
where the function $\tilde h$ is defined as
\begin{align*}
 \tilde h(y,a) & := 2M y^2\big(-1+2(M-1)y\big) + (1-My)\big(1-(M-2)y\big)a\\ 
& + 2\big(1-(M^2-1)y\big)a^2 + \big(1+(M-1)^2\big)a^3.   
\end{align*}
For $0<y<\frac{1}{2M^2}$ and $a>0$, the function $\tilde h$ satisfies
\begin{align}
&\frac{\partial \tilde h(y,a)}{\partial a}  = 3(1+(M-1)^2)a^2 + 4(1-(M^2-1)y)a+ (1-My)\big(1-(M-2)y\big)>0, \label{ha>0}\\
&\tilde h(y,\tfrac{M y^2}{(M+1)y+1}) = \frac{M y^2}{\big((M+1)y+1\big)^3} \Big( -1 - 2(3+M)y + (4M^2-13 -10M ) y^2 \nonumber \\
&\qquad\qquad\qquad\qquad + 4(M+1)^2 (2M-3) y^3 + (1+2M)(-4 + M^2 + 2M^3)y^4 \Big) <0. \label{ha<0}
\end{align}
The inequality in \eqref{ha>0} holds because $y<\frac{1}{2M^2}$. The inequality in \eqref{ha<0} follows from 
\begin{align*}
&-1 - 2(3+M)y + (-13 -10M + 4M^2) y^2 \\
&+ 4(M+1)^2 (2M-3) y^3 + (1+2M)(-4 + M^2 + 2M^3)y^4\\
&<-1 +4M^2 y^2 + 32 M^3 y^3 + 9 M^4 y^4 \\
&< - 1 + \frac{1}{M^2} + \frac{4}{M^3} + \frac{9}{16M^4}\\
&<0,
\end{align*}
where we have used $0<y<\frac{1}{2M^2}$ for the second inequality and $M\geq 2$ for the third inequality. Then, \eqref{th}, \eqref{ha>0}, and \eqref{ha<0} imply
\begin{align}
h(y+a,y,y-a) <0 \quad \textrm{for} \quad 0<a \leq \frac{M y^2}{(M+1)y+1} \quad \textrm{and}\quad 0<y<\frac{1}{2M^2}. \label{ha}
\end{align}

Since $0<\Sigma_n - \Sigma_{n+1} \leq \frac{M \Sigma_n^2}{(M+1)\Sigma_n+1} $,  we use \eqref{ha} with $a:=\Sigma_n - \Sigma_{n+1} $ to see
\begin{align}
h(2\Sigma_n - \Sigma_{n+1}, \Sigma_n,\Sigma_{n+1})<0 \quad \textrm{for}\quad  0<\Sigma_n <\frac{1}{2M^2}. \label{ha2}
\end{align}
Since $\Sigma_n^+ > 2\Sigma_n - \Sigma_{n+1}$ by \eqref{ys}, the inequalities in \eqref{h_negative} and \eqref{ha2} produce \eqref{difference_concavity}.
\ \\

\noindent 2. The lower bound $0 < N(\Sigma_0 - \Sigma_1)$ comes from the monotonicity of $\Sigma_n$. To see the upper bound in \eqref{derivative_bound}, we split into two cases. First, when $N=1$, the upper bound is trivial since $N(\Sigma_0 - \Sigma_1) = \Sigma_0 - \Sigma_1 < \Sigma_0$. 

Second, we assume $N \geq 2$. We set $D_n:=\Sigma_n- \Sigma_{n+1}$ for $0\leq n \leq N$ and note
\begin{align}
D_n>0, \quad \sum_{n=0}^N D_n =  \Sigma_0 - \Sigma_{N+1} = \Sigma_0. \label{D_n}
\end{align}
We define $n_0$ and $n^*$as
\begin{align}
n_0:= \min \left\{n\in \{1,2,\dots,N+1\}: \,\,\Sigma_n < \frac{1}{2M^2}   \right\}, \quad
n^* := \underset{0 \le n \le N}{\operatorname{argmin}}  \,\, D_n. \label{n_def}
\end{align}
Then, \eqref{D_n} implies
\begin{align}
(N+1)D_{n^*} \leq  \Sigma_0. \label{D_n0}
\end{align} 
For $1\leq n \leq N-1$, Lemma~\ref{lem_difference}.3 gives the first inequality in 
\begin{align}\label{D_ineq2}
\begin{split}
D_0 &=D_n \prod_{k=1}^n  \frac{D_{k-1}}{D_k}\\
&\leq D_n \prod_{k=1}^n  \frac{M^4 \Sigma_k^2 \Sigma_{k-1}\Sigma_{k+1}}{\big(\Sigma_k + (M-1)\Sigma_{k+1}\big)^4}\\
&\leq D_n \prod_{k=1}^n  \frac{M^4 \Sigma_k^2 \Sigma_{k-1}\Sigma_{k+1}}{\big(M \Sigma_k^{1/M} \Sigma_{k+1}^{(M-1)/M}\big)^4}\\
&= D_n \prod_{k=1}^n \Big(\frac{\Sigma_{k-1}}{\Sigma_k} \Big)  \Big(\frac{\Sigma_k}{\Sigma_{k+1}} \Big)^{\frac{3M-4}{M}} \\
&= D_n \frac{\Sigma_0 \Sigma_1^{(3M-4)/M}}{\Sigma_n \Sigma_{n+1}^{(3M-4)/M}} \\
&\leq D_n (\Sigma_0/\Sigma_{n+1})^{4(M-1)/M}. 
\end{split}
\end{align}
The second inequality in \eqref{D_ineq2}  is due to the AM-GM inequality which gives
$$
\Sigma_k +  \underbrace{\Sigma_{k+1} + ... + \Sigma_{k+1}}_{M-1 \text{ terms}} \ge M \Big( \Sigma_k \big(\Sigma_{k+1}\big)^{M-1}\Big)^\frac1M = M \ \Sigma_k^\frac1M \big(\Sigma_{k+1}\big)^{\frac{M-1}M}.
$$
The last inequality in \eqref{D_ineq2} comes from  $\Sigma_n$ being decreasing in $n$.

We split the argument into cases:

\noindent {\bf Case 1/4:} Suppose that $n^*=0$. Then,   \eqref{derivative_bound} holds by \eqref{D_n0}.

\noindent {\bf Case 2/4:} Suppose that $n_0=1$. Then, \eqref{difference_concavity} implies that $D_0<D_1<\cdots<D_N$. Therefore, $n^*=0$ and \eqref{derivative_bound} holds by Case 1.

\noindent {\bf Case 3/4:} Suppose that $n^*\geq 1$ and $2\leq n_0 \leq N$. Then, \eqref{n_def} gives $\Sigma_{n_0}<\frac{1}{2M^2}\leq \Sigma_{n_0 - 1}$. The concavity in  \eqref{difference_concavity} implies $\Sigma_{n_0 - 1} = \sum_{n=n_0 -1}^N D_n > (N-n_0 +2) D_{n_0 -1}$ and $n^*\leq n_0-1$. These observations produce
\begin{align}
\Sigma_{n^*+1}\geq \Sigma_{n_0} = \Sigma_{n_0 - 1} - D_{n_0 -1} > \frac{N-n_0+1}{N-n_0+2}\,\, \Sigma_{n_0 - 1} \geq \frac{1}{4M^2}. \label{Sigma_n0}
\end{align}
By combining \eqref{D_n0}, \eqref{D_ineq2}, and \eqref{Sigma_n0}, we obtain
\begin{align*}
N D_0 \leq N D_{n^*} \, (\Sigma_0/\Sigma_{n^*+1})^{4(M-1)/M} \leq  (4M^2)^{4(M-1)/M} \Sigma_0^{(5M-4)/M} .
\end{align*}

\noindent {\bf Case 4/4:} Suppose that $n^*\geq 1$ and $n_0=N+1$. Because $\Sigma_{N+1}=0$, $n_0=N+1$ gives $D_N=\Sigma_N\geq \frac{1}{2M^2}$ and the assumption $\frac{N+1}{2M^2} > \Sigma_0$ implies that  $n^* < N$. Indeed, if $n^*=N$, we would get the contradiction
$$
\Sigma_0 = \sum_{n=0}^N D_n \ge  \sum_{n=0}^N D_N \ge \frac{(N+1)}{2M^2}.
$$
Therefore, $1\leq n^* \leq N-1$. Because $\Sigma_n$ is decreasing, we have $\Sigma_{n^* +1} \geq \Sigma_N \geq \frac{1}{2M^2}$ and the inequalities in \eqref{D_n0} and \eqref{D_ineq2} produce
\begin{align*}
N   D_0 \leq N D_{n^*} \, (\Sigma_0/\Sigma_{n^*+1})^{4(M-1)/M} \leq  (2M^2)^{4(M-1)/M} \Sigma_0^{(5M-4)/M} .
\end{align*}

$\endproof$

\subsection{ODE estimates}

The next result collects properties of the ODE in \eqref{BCW1}, which we need for our convergence proof.
 To simplify the exposition, we normalize the initial value to $S(0)=1$. This assumption is equivalent to considering $\rho := \frac{(M-2) }{2 (M-1)}\sigma_{\ta}^2$ in what follows and is without any loss of generality.

\begin{lemma}\label{lemma:BCWplus} For $M\ge 2$ partially informed traders and $\rho := \frac{(M-2) }{2 (M-1)}\sigma_{\ta}^2$, we have
\begin{enumerate}

\item Uniqueness holds for \eqref{BCW1} and the solution satisfies $S'(1^-) =-\infty$.

\item For any constant $b <0$, the first-order ODE 
\begin{align}\label{BCW2}
\begin{cases}
S_b'(t) =  b S_b(t)^{4-\frac4M} e^{\frac{2}{MS_b(t)} - \frac2M},\quad t>0, \\
S_b(0) = 1,
\end{cases}
\end{align}
has a unique local solution. The maximal interval of existence is $[0,T_b)$ where
\begin{align}\label{TbandA}
T_b := -\frac{A}{b} e^{\frac2M},\quad A:= \int_0^1 u^{\frac4M -4}e^{-\frac2{Mu}} du \in (0,\infty).
\end{align}
Furthermore, $ \lim_{t\uparrow T_b} S_b(t)=0$.

\item For $b < S'(0)$, there exists $\delta \in (0,\frac1{4M^2})$ such that
\begin{align}\label{2conds_delta}
T_b < 1-2\delta,\quad S'_b (\tau(b,\delta)) <-1,
\end{align}
where, for $b <0$ and $\delta \in (0,1)$, the first passage time is
\begin{align}\label{tau_b_delta}
\tau(b,\delta) := \inf\{t>0 : S_b(t) \le \delta\} \in (0,T_b).
\end{align}

\end{enumerate}
\end{lemma}

\proof 

1. We define the function
$$
H(s) :=  \int_s^1 u^{\frac4M -4}e^{-\frac2{Mu}} du,\quad s\in[0,1].
$$
This function is positive and strictly decreasing. Because the ODE \eqref{BCW1} gives
$$
H\big(S(t)\big)'' = 0,
$$
the function $t\to H(S(t))$ is affine. Let $A$ be as in \eqref{TbandA}. The two properties
$$
H\big(S(0)\big) = H(1) =0,\quad H\big(S(1)\big) = H(0) = A
$$
give $H\big(S(t)\big) = At$. Therefore,
$$
A = H\big(S(t)\big)' = H'\big(S(t)\big)S'(t).
$$
This gives (below, we need the limit at $t=0$)
\begin{align*}
&\lim_{t\uparrow 1} S'(t) =\lim_{t\uparrow 1}  \frac{A}{H'\big(S(t)\big)} = -\infty,\\
&\lim_{t\downarrow 0} S'(t) =\lim_{t\downarrow 0}  \frac{A}{H'\big(S(t)\big)} =   \frac{A}{H'(1)} = - A e^{\frac2M} .
\end{align*}

To see uniqueness, we let $S$ and $\tilde S$ be two solutions of  \eqref{BCW1}. The transformed functions
\begin{align}\label{uniquenessY}
Y(t) := H\big(S(t)\big),\quad \tilde Y(t) := H\big(\tilde S(t)\big)
\end{align}
both equal $At$. Because $H$ is one-to-one, we can invert to see that uniqueness holds. \ \\

\noindent 2.  Because 
$$
H\big(S_b(t)\big)' = - b e^{-\frac2M}, 
$$
the function $t\to H\big(S_b(t)\big)$ is affine. The initial condition $S_b(0)=1$ gives
$$
H\big(S_b(0)\big) = H(1) =0
$$
and so $H\big(S_b(t)\big) = - b e^{-\frac2M} t$. Because $H$ is one-to-one, we see that $S_b(T_b) =0$ if and only if
$$
A= H(0) = H\big(S_b(T_b)\big) = - b e^{-\frac2M} T_b.
$$

Uniqueness is shown by replacing $S$ and $\tilde S$ in \eqref{uniquenessY} above with two solutions $S_b$ and $\tilde S_b$ and then use $H\big(S_b(t)\big) =H(\tilde S_b(t)) = - b e^{-\frac2M} t$ together with $H$ being one-to-one.

\noindent  3.  Because $b < S'(0) <0$ we have
$$
T_b =  -\frac{A}{b} e^{\frac2M} < -\frac{A}{S'(0)} e^{\frac2M} = -H'(1)  e^{\frac2M}=1.
$$ 
By definition of $\tau(b,\delta)$ in \eqref{tau_b_delta}, we have $S_b(\tau(b,\delta))=\delta$ and so
\begin{align*}
S_b'(\tau(b,\delta)) &=  b \delta^{4-\frac4M} e^{\frac{2}{M\delta} - \frac2M}\\
&<  S'(0) \delta^{4-\frac4M} e^{\frac{2}{M\delta} - \frac2M}\\
&= -A \delta^{4-\frac4M} e^{\frac{2}{M\delta}}.
\end{align*}
Therefore, 
$$
\lim_{\delta \downarrow 0}S_b'(\tau(b,\delta)) = -\infty,
$$
which allows us to find arbitrary small values of $\delta>0$ such that \eqref{2conds_delta} holds.
$\endproof$

\subsection{Convergence of recursion to ODE}
Unlike the proof of Lemma \ref{lem_difference}, we run the recursion from $n=0$ to $n=N$ throughout this section. 
For $0\le z\le y\le x$, the function $h = h(x,y,z)$ defined in \eqref{h_def} satisfies
\begin{align}\label{h_forward}
\begin{split}
h(x,y,0)&=y \, g(x,y)^2,  \\
h(x,y,y)&=-M^2 x (x-y) y^3,\\
h_3(x,y,0)&=-g(x,y)^2 - x(x-y)y^2,\\
h_{33}(x,y,z) &= -2(M-1) x(x-y) (2y + 3(M-1)z),
\end{split}
\end{align} 
where $h_3$ and $h_{33}$ denote the partial derivatives of $h$ with respect to its third argument.
The above expressions imply that for $x\geq y \geq 0$, there exists a unique $z\in [0,y]$ such that $h(x,y,z)=0$.
Therefore, for $b\in (-\infty,0)$ and $N\in\N$, we can uniquely determine $\{\Sigma_{N,b}(n)\}_{n=0}^{N+1}$ as the solution of the following difference equations:
\begin{align}\label{Sigma_forward}
\begin{split}
&\Sigma_{N,b}(0) = 1, \\  
&\Sigma_{N,b}(1) = 1+ b/N,\\
&h\big(\Sigma_{N,b}(n-1), \Sigma_{N,b}(n), \Sigma_{N,b}(n+1)\big)=0 \quad \textrm{for}\quad n=1,2,...,N.
\end{split}
\end{align}
If the solution $\{\Sigma_{N,b}(n)\}_{n=0}^{N+1}$ also satisfies 
\begin{align}\label{extra2conds}
\Sigma_{N,b}(N+1)=0,\quad 0<\Sigma_{N,b}(n)<\Sigma_{N,b}(n-1)\leq \Sigma_{N,b}(n)^+,\quad n=1,2,...,N,
\end{align} 
then $\left\{\Sigma_{N,b}(n)\right\}_{n=0}^N$ is the solution of the original backward recursion in Lemma~\ref{lem_difference}.
The connection between the forward recursion $\Sigma_{N,b_N}(n)$ from \eqref{Sigma_forward} and the backward recursion $\Sigma_n$ from  Lemma~\ref{lem_difference} can be seen by setting $b_N:= N(\Sigma_1-\Sigma_0) \in (-\infty,0)$, in which case $\Sigma_{N,b_N}(n)=\Sigma_n$.

The next lemma considers the forward recursion and neither of the two conditions in \eqref{extra2conds} are imposed. Even though the lemma's proof is long, our arguments are fairly standard for proving convergence of a discrete scheme to the solution of a second-order ODE. For $b<0$ and $\delta>0$, $T_b$, $\tau(b,\delta)$, and $S_b$ are from Lemma \ref{lemma:BCWplus}. For $b=0$ and $\delta >0$, we define $T_b:=\tau(b,\delta):=\infty$. Furthermore, using $b=0$ in \eqref{BCW2} gives $S_0(t) = 1$ for all $t\ge0$.


\begin{lemma}\label{difference_to_ODE}  For $M\ge 2$ partially informed traders and $\rho := \frac{(M-2) }{2 (M-1)}\sigma_{\ta}^2$, we assume $\lim_{N\to \infty} b_N = c\leq 0$. Then, for $\delta\in (0,1)$ and $t\in [0,\tau(c,\delta)\wedge 1]$, we have
\begin{align}
\lim_{N\to \infty}\Sigma_{N,b_N}(\floor{tN}) = S_c(t) \quad \textrm{and}\quad 
\lim_{N\to \infty}\tfrac{\Sigma_{N,b_N}(\floor{tN}+1)-\Sigma_{N,b_N}(\floor{tN})}{1/N} = S'_c(t). \label{difference_limits}
\end{align}
\end{lemma}
\proof
Throughout this proof, $\overline C$ and $ \underline C >0$ are generic constants that do not depend on $n$ and $N$ and may differ from line to line. We abbreviate by writing $\tau=\tau(c,\delta)$ and
\begin{align}
u_n:=\Sigma_{N,b_N}(n), \quad v_n:= N(u_{n+1}-u_n),\quad t_n:=\tfrac{n}{N}.
\end{align}
Then, $1=u_0>u_1>\dots$ is decreasing in $[0,1]$ and $v_n<0$ for all $n\in \{0,1,...,N\}$.

Let $\eta\in (\tau\wedge 1, T_c)$ be arbitrary for $T_c$ defined in \eqref{TbandA}. Lemma \ref{lemma:BCWplus}.2 gives 
\begin{align}
&m:=\min_{t\in [0,\eta]} S_c(t) = S_c(\eta)>0,\\
&S_c''(t)=F(S_c(t))\,S_c'(t)^{2}, \quad
   F(s):=\tfrac{4(M-1)s-2}{Ms^{2}}.
\end{align}
We define $C_v\in (0,\infty)$ by
\begin{align}
C_v:= \tfrac{m}{4}+ \max_{t\in [0,\eta]} | S_c'(t)| . \label{Cv_def}
\end{align}
We can choose large enough $\overline C>0$ such that
\begin{align}
\max_{t\in [0,\eta]} | S_c'(t)| \vee |S_c''(t)| \leq \overline C, \quad \max_{s\in [\frac{m}{2},1]} |F(s)| \vee |F'(s)| \leq \overline C. \label{bdd}
\end{align}
By Taylor's theorem and \eqref{bdd}, for $n\ge1$ with $t_n\leq \eta$, we have
\begin{align}
&\Big|S_c(t_n)- S_c(t_{n-1})- \tfrac{1}{N} S'_c(t_{n-1}) \Big| \leq \tfrac{\overline C}{N^2}, \label{ST1}\\
&\Big|S_c'(t_n)- S_c'(t_{n-1})- \tfrac{1}{N} F\big(S_c(t_n)\big) S'_c(t_{n-1})^2 \Big| \leq \tfrac{\overline C}{N^2}. \label{ST2}
\end{align}
\medskip
\noindent\textbf{Claim 1.} \emph{There exists $N_1\in \mathbb{N}$ with the following property: If $u_n \geq \frac{m}{2}$ and $v_{n-1}\geq -C_v$ for some $N\ge N_1$ and some $n \in \{1,...,N\}$, then }
\begin{align}
\left | v_n - v_{n-1} - \tfrac{1}{N}F(u_n) v_{n-1}^2  \right| \leq \tfrac{\overline C}{N^2}. \label{vn_claim}
\end{align}

\noindent\textit{Proof of Claim 1.}\quad We check by direct computations that $h(y+p,y,y-p+F(y)p^2)$ is a polynomial in $p$, divisible by $p^3$, with bounded coefficients for $y\in [\frac{m}{2},1]$. 
Therefore,
\begin{align}
\big| h(y+p,y,y-p+F(y)p^2) \big| \leq \overline C p^3 \quad \textrm{for} \quad p\in [0,1] \textrm{  and  } y\in [\tfrac{m}{2},1]. \label{h_approx_sol}
\end{align}
We check by direct computations that $h_3(y+p,y,y-q) + M^2 y^4$ is a polynomial in $p$ with no constant terms. Hence, there exist $\underline C>0$ and $p_*\in (0,1]$ such that
\begin{align}
\left.\begin{aligned}
&p-F(y)p^2\in [0,\tfrac{3}{2}p]\\
&  h_3(y+p,y,y-q) \leq - \underline C
\end{aligned} \,\,\, \right\} 
 \quad \textrm{for}\quad p\in [0,p_*], \,\, q\in [0,2p], \,\, y\in[\tfrac{m}{2},1]. \label{h_increase_q}
\end{align}
The first property in \eqref{h_increase_q} implies $p-F(y)p^2 < 2p$ so there exists $\theta \in (p-F(y)p^2,2p)$ such that
\begin{align*}
h(y+p,y,y-2p) -h(y+p,y,y-p+F(y)p^2) &=-h_3(y+p,y,y-\theta) p\big(1+pF(y)\big) \\
&\ge \underline C \,\,  p(1+pF(y)),
\end{align*}
where the inequality is from the second property in \eqref{h_increase_q}. Combining this with the bound in \eqref{h_approx_sol} gives 
\begin{align*}
h(y+p,y,y-2p) \ge \underline C p\big(1+pF(y)\big)-\overline C p^3.
\end{align*}
Therefore, by reducing $p_*$ if needed, we can assume $p_* \in (0,1]$ satisfies \eqref{h_increase_q} as well as
\begin{align}
h(y+p,y,y-2p)>0 \quad\textrm{for}\quad  p\in (0,p_*] \,\, \textrm{ and } \,\,y\in[\tfrac{m}{2},1]. \label{h(2p)>0}
\end{align}

By \eqref{h_increase_q}, the map $q\mapsto h(y+p,y,y-q)$ is strictly increasing on $[0,2p]$ with $h(y+p,y,y)=-M^2 (p+y)p y^3<0$. This observation and \eqref{h(2p)>0} imply that a unique root $Q(y,p)\in (0,2p)$ of the function $(0,2p)\ni q \to h(y+p,y,y-q)$ exists in the sense
\begin{align}
\left.\begin{aligned}
&h\big(y+p,y,y-Q(y,p)\big)=0\\
&Q(y,p)\in (0, 2p)
\end{aligned} \,\,\, \right\} 
\quad \textrm{for}\quad p\in (0,p_*] \,\, \textrm{ and } \,\,y\in[\tfrac{m}{2},1]. \label{Q_def}
\end{align}

To see Claim 1, we define $N_1:= \lceil \frac{C_v}{p_*}\rceil$.  Assume $u_n \geq \frac{m}{2}$ and $- C_v \leq v_{n-1}=N(u_n - u_{n-1})$ for some $N\ge N_1$ and $n \in \{1,...,N\}$. Since $1=u_0>...>u_N>0$, we have
\begin{align}
y:=u_n \in [\tfrac{m}{2}, 1] \quad \textrm{and} \quad p:= -\tfrac{v_{n-1}}{N}\in (0,p_*]. \label{uv_bound}
\end{align}
Since $0=h(u_{n-1},u_n,u_{n+1})=
h(u_n - \tfrac{v_{n-1}}{N}, u_n,  u_n + \tfrac{v_n}{N})$, \eqref{Q_def} and \eqref{uv_bound} imply
$$
q:= - \tfrac{v_n}{N}= Q(u_n,  - \tfrac{v_{n-1}}{N} ) \in (0, -\tfrac{2v_{n-1}}{N}).
$$  
If $q=p-F(y)p^2$,  there is nothing to prove. When $q\neq p-F(y)p^2$, we have 
\begin{align}
 \tfrac{\overline C}{N^3} &\geq \big| h(y+p,y,y-p+F(y)p^2) \big| \nonumber\\
 &= \big| h\big(y+p,y,y-p+F(y)p^2\big) - h(y+p,y,y-q) \big| \nonumber\\
&=\big|  h_3(y+p,y,y-\theta) \big| \big| q -p + F(y)p^2 \big| \nonumber \\
&\geq \underline C \,\, \big|  \tfrac{v_n}{N} - \tfrac{v_{n-1}}{N} - F(u_n) \tfrac{v_{n-1}^2}{N^2} \big|,  \label{mvt1}
\end{align}
 where the first inequality is from \eqref{h_approx_sol}, the first equality uses $h(y+p,y,y-q)=0$, the second equality is produced by the Mean-Value Theorem, which gives $\theta$ between $q$ and $p-F(y)p^2$, and  the last inequality is from \eqref{h_increase_q}. This finishes the proof of Claim 1. 
$\endproof$

\medskip

Next, we combine Claim~1 with \eqref{ST1} and \eqref{ST2} to control the error
\begin{align}
E_n:=|u_n-S_c(t_n)|+|v_n-S_c'(t_n)|.
\end{align}

\medskip
\noindent\textbf{Claim 2.} \emph{There exist $C_E>0$ and $N_2\ge N_1$ such that for all $N\geq N_2$ and all $n\in \{0,...,N\}$ with $t_n \leq \eta$, the following holds:}
\begin{align}
u_n\geq \tfrac m2,\quad v_n\geq -C_v,\quad E_n\leq C_E \big(|b_N-c|+\tfrac1N\big). \label{Gronwall_bd}
\end{align}
\noindent\textit{Proof of Claim 2.}

{\bf Step 1/2:}  Suppose that $N\geq N_1$, $u_n \in [\frac m2,1]$, and $v_{n-1}\in [-C_v,0)$. From \eqref{ST1}, we obtain
\begin{align}
|u_n-S_c(t_n)| &= \big| u_{n-1} - S_c(t_{n-1}) + \tfrac{1}{N}\big(v_{n-1} -S_c'(t_{n-1})\big) - \big( S_c(t_n) - S_c(t_{n-1}) - \tfrac{1}{N}S_c'(t_{n-1}) \big) \big| \nonumber \\
&\leq |u_{n-1}-S_c(t_{n-1})|+\tfrac1N|v_{n-1}-S_c'(t_{n-1})|+\tfrac{\overline C}{N^2}. \label{E_bdd1}
\end{align}
The bounds in \eqref{bdd}, $u_n\in [\tfrac{m}{2},1]$, and $ v_{n-1}\in [ -C_v,0)$ imply that 
\begin{align}
&\big| F(u_n)v_{n-1}^2-F(S_c(t_n))S_c'(t_{n-1})^2 \big| \nonumber\\
&=\Big| \Big(F(u_n)-F\big(S_c(t_n)\big)\Big)v_{n-1}^2+F\big(S_c(t_n)\big)\big(v_{n-1}^2-S_c'(t_{n-1})^2\big) \Big | \nonumber\\
&\leq \overline C  \big(  | u_n - S_c(t_n)| + |v_{n-1} - S_c'(t_{n-1})|\big), \label{mixed_bdd}
\end{align}
 where the inequality uses the relation $(a^2-b^2) = (a+b)(a-b)$ for $a,b\in\R$. We combine \eqref{ST2}, \eqref{vn_claim}, and \eqref{mixed_bdd} to see
\begin{align}
|v_n - S_c'(t_n)| &= \Big| \big( v_n - v_{n-1} - \tfrac{1}{N} F(u_n) v_{n-1}^2 \big)  - \Big( S_c'(t_n) - S_c'(t_{n-1}) - \tfrac{1}{N} F\big(S_c(t_n)\big) S_c'(t_{n-1})^2 \Big)  \nonumber\\
&\qquad + \big( v_{n-1} - S_c'(t_{n-1}) \big)+ \tfrac{1}{N}\Big( F(u_n)v_{n-1}^2 - F\big(S_c(t_{n})\big)S_c'(t_{n-1})^2 \Big) \Big|  \nonumber\\
&\leq \tfrac{\overline C}{N^2}+(1+ \tfrac{\overline C}{N}) \, |v_{n-1}- S_c'(t_{n-1})| +\tfrac{\overline C}{N} \, | u_n - S_c(t_n)| . \label{E_bdd2}
\end{align}
By increasing $\overline C$, \eqref{E_bdd1} and \eqref{E_bdd2} give the implication
\begin{align}
 (u_n,v_{n-1})\in [\tfrac{m}{2},1]\times [-C_v,0) \quad \textrm{and} \quad N\geq N_1 \quad \Longrightarrow\quad E_n \leq (1+\tfrac{\overline C}{N}) E_{n-1} +\tfrac{\overline C}{N^2} . \label{E_bdd3}
\end{align}

{\bf Step 2/2:}  Let $C_*$ be a fixed constant greater than all the $\overline C$ appearing in \eqref{E_bdd1}-\eqref{E_bdd3}. 
Since $b_N\to c$, we can choose $N_2\ge N_1$ such that
\begin{align}
\forall N\geq N_2: e^{\eta C_* }\big(|b_N-c|+\tfrac{1}{N}\big) <\tfrac{m}{4} \quad\text{and}\quad \tfrac{C_*}{N^2}<\tfrac{m}{4}. \label{E_bdd4}
\end{align}

We prove \eqref{Gronwall_bd} for $C_E:=e^{ \eta C_* } $ using induction in $n$. 
For $n=0$, we have $u_0=1\geq \frac{m}{2}$ and 
$E_0=|v_0 -c | =|b_N-c|<\tfrac{m}{4}$.
This inequality gives 
\begin{align*}
-C_v= -\tfrac{m}{4}- \max_{t\in [0,\eta]} | S_c'(t)|< -\tfrac{m}{4} - | S'_c(0) | = -\tfrac{m}{4} +c < v_0.
\end{align*}

For the induction step, we suppose that \eqref{Gronwall_bd} holds for $0,...,n-1$, where $n\in \{1,..., N\}$ and $t_n\leq \eta$. Part of the induction hypothesis gives $E_{n-1} \le C_E \big(|b_N-c|+\tfrac1N\big)$, which is bounded by $\frac{m}4$ by the first part of  \eqref{E_bdd4}. This together with \eqref{E_bdd1} and the second part of \eqref{E_bdd4} yields 
\begin{align}
|u_n-S_c(t_n)| \leq E_{n-1}+ \tfrac{C_*}{N^2}<\tfrac{m}{4}+\tfrac{m}{4}=\tfrac{m}{2}.
\end{align}
Therefore,
$$
u_n >S_c(t_n)- \tfrac{m}{2} \geq \tfrac{m}{2},
$$
which is the first property in  \eqref{Gronwall_bd}. 

Because $v_{n-1} \ge - C_v$ by the induction hypothesis, the recursive inequality in \eqref{E_bdd3} and $C_*\geq \overline C$ produce
\begin{align}\label{E_bdd5}
\begin{split}
E_n &\leq (1+\tfrac{C_*}{N}) E_{n-1} +\tfrac{C_*}{N^2} \\
&\leq (1+\tfrac{C_*}{N})^n | b_N-c |  + \tfrac{C_*}{N^2} \sum_{k=0}^{n-1} (1+\tfrac{C_*}{N})^k \\
&\leq e^{\eta C_* } \big( | b_N- c| +  \tfrac{1}{N}\big),
\end{split}
\end{align}
where the last inequality uses $1+x \leq  e^x$ and $t_n\leq \eta$. This is the last property in  \eqref{Gronwall_bd}. 

Finally, the first part of \eqref{E_bdd4} and \eqref{E_bdd5} imply
\begin{align*}
|v_n-S_c'(t_n)|\le E_n<\tfrac{m}{4}.
\end{align*}
The definition of $C_v$ gives
$$
 v_n> S_c'(t_n) - \tfrac{m}{4} \geq -C_v,
$$
which is the second property in  \eqref{Gronwall_bd}. This finishes the proof of Claim 2. 
$\endproof$
\medskip

To prove Lemma \ref{difference_to_ODE}, we fix $t\in[0,\tau\wedge 1]$ and set $n:=\floor{tN}$. Since $t\le 1$, we have $n\le N$, so $v_n:=N(u_{n+1}-u_n)$ is well defined, and $t_n=\tfrac nN\le t\le \tau \wedge 1 < \eta$. Hence, by Claim 2 and \eqref{bdd}, for $N\geq N_2$, 
\begin{align*}
| u_n - S_c(t) | + |v_n-S_c'(t)| &\leq | u_n - S_c(t_n) | + |v_n-S_c'(t_n)| + | S_c(t_n) - S_c(t) | + | S_c'(t_n) - S_c'(t) | \\
&\leq  \overline C \big(|b_N-c|+\tfrac1N\big) \xrightarrow[N\to\infty]{}0,
\end{align*}
which gives both limits in \eqref{difference_limits}. 

$\endproof$

In the next lemma, the constant $A$ is from \eqref{TbandA} and $\Sigma_0,...,\Sigma_N$ is from  Lemma~\ref{lem_difference}.

\begin{lemma}\label{bN_limit} In the setting of Theorem \ref{thm_Main}, let $b_N:= N(\Sigma_1-\Sigma_0)$. Then, we have $\lim_{N\to \infty} b_N = - A\, e^{\frac{2}{M}}$.
\end{lemma}
\proof
Lemma~\ref{bN_bound}.2 gives 
$$
\bar c:= \limsup_{N\to \infty} b_N \in (-\infty,0], \quad \underline c:= \liminf_{N\to \infty} b_N\in (-\infty,0].
$$

{\bf Step 1/2:}  We set $b_\infty := - A\, e^{\frac{2}{M}}<0$ and suppose $\bar c > b_\infty $. 
Since $-\bar c < -b_\infty $, we have from \eqref{TbandA} that $T_{\bar c} = \frac{A}{-\bar c}e^{\frac{2}{M}} > 1$ for $\bar c <0$ and we use $\tau(0, \delta):=T_{0}:=\infty$ for $\bar c=0$. Therefore, $S_{\bar c}(1) > 0$ and so we can choose $\delta \in \big(0, S_{\bar c}(1)\big)$and $\tau(\bar c, \delta) \wedge 1 = 1$. 
Let $\{b_{N_k}\}_{k\in \N}$ be a subsequence such that $\lim_{k\to \infty} b_{N_k} =\bar c$.
By Lemma~\ref{difference_to_ODE}, 
\begin{align*}
\lim_{k\to \infty}\Sigma_{N_k,b_{N_k}}(N_k) = S_{\bar c}(1)>0.
\end{align*}
However, this limit and the second limit in Lemma~\ref{difference_to_ODE} yield a contradiction:
\begin{align*}
-\infty =\lim_{k\to \infty}N_k(-\Sigma_{N_k,b_{N_k}}(N_k))=\lim_{k\to \infty}\frac{\Sigma_{N_k,b_{N_k}}(N_k+1)-\Sigma_{N_k,b_{N_k}}(N_k)}{1/N_k} = S'_{\bar c}(1) >-\infty.
\end{align*}

All in all, $\bar c > b_\infty$ is impossible and we conclude that  $\bar c \le b_\infty$.\ \\

{\bf Step 2/2:} Suppose that $\underline c< b_\infty$. By Lemma~\ref{lemma:BCWplus}.3, there exists $\delta\in (0,\frac{1}{4M^2})$ such that 
\begin{align}
T_{\underline c}<1-2\delta \quad \textrm{and} \quad S_{\underline c}'(\tau(\underline c,\delta))< -1. \label{cannot_reach}
\end{align}
Let $\{b_{N_k}\}_{k\in \N}$ be a subsequence such that $\lim_{k\to \infty} b_{N_k} =\underline c$.
By Lemma~\ref{difference_to_ODE}, we have
\begin{align*}
&\lim_{k\to \infty}\Sigma_{N_k,b_{N_k}}(\floor{\tau(\underline c, \delta)N_k}) = S_{\underline c}\big(\tau(\underline c, \delta)\big)=\delta,\\
&\lim_{k\to \infty}\frac{\Sigma_{N_k,b_{N_k}}(\floor{\tau(\underline c, \delta)N_k}+1)-\Sigma_{N_k,b_{N_k}}(\floor{\tau(\underline c, \delta)N_k})}{1/N_k} = S'_{\underline c}\big(\tau(\underline c, \delta)\big)<-1.
\end{align*}
These limits imply that for large enough $k$,
\begin{align}
&\Sigma_{N_k,b_{N_k}}(n) <2\delta<\tfrac{1}{2M^2} \quad \textrm{for} \quad n=\floor{\tau(\underline c, \delta)N_k},..., N_k, \label{cannot_reach2}\\
&\Sigma_{N_k,b_{N_k}}(\floor{\tau(\underline c, \delta)N_k})- \Sigma_{N_k,b_{N_k}}(\floor{\tau(\underline c, \delta)N_k}+1) >\tfrac{1}{N_k}. \nonumber
\end{align}
We fix such $k$ satisfying the above inequalities.
Lemma~\ref{bN_bound}.1 and the above inequalities produce
\begin{align}
\Sigma_{N_k,b_{N_k}}(n) - \Sigma_{N_k,b_{N_k}}(n+1)>\tfrac{1}{N_k}  \quad \textrm{for} \quad n=\floor{\tau(\underline c, \delta)N_k},..., N_k. \label{cannot_reach3}
\end{align}
Then, \eqref{cannot_reach2} and \eqref{cannot_reach3} produce the contradiction:
\begin{align*}
2\delta &> \Sigma_{N_k,b_{N_k}}(\floor{\tau(\underline c, \delta)N_k}) \\
&= \sum_{n=\floor{\tau(\underline c, \delta)N_k}}^{N_k} \Big(\Sigma_{N_k,b_{N_k}}(n) - \Sigma_{N_k,b_{N_k}}(n+1)\Big)\\
& > \frac{ N_k - \floor{\tau(\underline c, \delta)N_k} + 1 }{N_k} \\
&> 2\delta,
\end{align*}
where the last inequality is from $\tau(\underline c, \delta)< T_{\underline c}<1-2\delta$ by \eqref{cannot_reach}.  All in all, $\underline c < b_\infty$ is impossible and we conclude that  $\underline c \ge b_\infty$.\ \\

\smallskip

Steps 1 and 2 imply that $\bar c \leq b_\infty \leq \underline c$, and we conclude that $\bar c = \underline c=b_\infty$.
$\endproof$

\bigskip

\proof[Proof of Theorem \ref{thm_converge}]
 Lemma~\ref{bN_limit} gives $\lim_{N\to\infty} b_N = b_\infty := - A\, e^{\frac{2}{M}}<0$. From Lemma \ref{lemma:BCWplus}, we see $S = S_{b_\infty}$ with 
$T_{b_\infty} = 1$. Fix $t \in [0,1)$. Since $\tau(b_\infty, \delta) \uparrow 1$ as 
$\delta \downarrow 0$, we may choose $\delta \in (0,1)$ with $t < \tau(b_\infty, \delta)$. 
Then, Lemma~\ref{difference_to_ODE} produces \eqref{convergence_S_Sprime}.
$\endproof$

\end{document}